\documentclass[11pt]{article}
\usepackage{enumerate}
\usepackage{natbib}
\usepackage{amsmath,amsthm,amssymb}

\usepackage{authblk}{\tiny }

\theoremstyle{plain}

\theoremstyle{definition} 

\newtheorem{remark}{Remark}

\usepackage{graphicx,psfrag,epsf}
\usepackage{epstopdf}
\usepackage{epsfig}
\usepackage[toc]{appendix}
\usepackage{bbm}
\usepackage{enumitem}
\usepackage{xcolor}
\usepackage{url} 
\usepackage{algorithm}
\usepackage{algpseudocode}
\usepackage{bm}
\usepackage{bigints}
\usepackage{mathrsfs}

\usepackage{multirow, array, hhline}

\usepackage{subcaption}

\usepackage{enumitem}

\usepackage{algorithm,algpseudocode}
\algnewcommand{\Inputs}[1]{%
  \State \textbf{Inputs:}
  \Statex \hspace*{\algorithmicindent}\parbox[t]{.8\linewidth}{\raggedright #1}
}
\algnewcommand{\Initialize}[1]{%
  \State \textbf{Initialize:}
  \Statex \hspace*{\algorithmicindent}\parbox[t]{.8\linewidth}{\raggedright #1}
}

\usepackage{graphicx}

\usepackage{datatool}
\usepackage{booktabs} 

\RequirePackage[hyperindex,breaklinks,colorlinks,citecolor=blue,urlcolor=black]{hyperref}

\theoremstyle{plain}
\newtheorem{theorem}{Theorem}
\newtheorem{lemma}{Lemma}

\newcommand\iid{\stackrel{\rm iid}{\sim}}
\newcommand\J{{\mathcal J}}

\newcommand{\blind}{0}

\begin{document}

   \def\spacingset#1{\renewcommand{\baselinestretch}%
      {#1}\small\normalsize} \spacingset{1}

   
   \if0\blind
   {
      \title{ Penalty-Induced Basis Exploration for Bayesian Splines
  }  
      \author[1]{Sunwoo Lim}
      \author[2]{Sihyeon Pyeon}
      \author[2,3]{Seonghyun Jeong\thanks{Corresponding author: sjeong@yonsei.ac.kr}} 
      \affil[1]{Marshall School of Business, University of Southern California, Los Angeles, CA, USA}
      \affil[2]{Department of Statistics and Data Science, Yonsei University, Seoul, Korea}
      \affil[3]{Department of Applied Statistics, Yonsei University, Seoul, Korea}
      \maketitle
   } \fi
   
   \if1\blind
   {
      \bigskip
      \bigskip
      \bigskip
      \begin{center}
         {\LARGE\bf Title}
      \end{center}
      \medskip
   } \fi
   
   \begin{abstract}
 Spline basis exploration via Bayesian model selection is a widely employed strategy for determining the optimal set of basis terms in nonparametric regression. However, despite its widespread use, this approach often encounters performance limitations owing to the finite approximation of infinite-dimensional parameters. This limitation arises because Bayesian model selection tends to favor simpler models over more complex ones when the true model is not among the candidates. Drawing inspiration from penalized splines, one potential remedy is to incorporate an additional roughness penalty that directly regulates the smoothness of functions. This strategy mitigates underfitting by allowing the inclusion of more basis terms while preventing overfitting through explicit smoothness control. Motivated by this insight, we propose a novel penalty-induced prior distribution for Bayesian basis exploration. The proposed prior evaluates the complexity of spline functions based on a convex combination of a roughness penalty and a ridge-type penalty for model selection. Our method adapts to the unknown level of smoothness and achieves the minimax-optimal posterior contraction rate up to a logarithmic factor. We also provide an efficient Markov chain Monte Carlo algorithm for its implementation. Extensive simulation studies demonstrate that our method outperforms competing approaches in terms of performance metrics and model complexity. An application to real datasets further substantiates the validity of our proposed approach.
   \end{abstract}
   
   \noindent%
   
   \spacingset{1.1} 
   

   \section{Introduction}\label{sec:intro}

We consider a Bayesian framework for analyzing the multivariate nonparametric regression model,
    \begin{align} 
    	\label{eqn:nonparamreg}
        y_i = f(x_i) + \varepsilon_i, \quad  \varepsilon_i \iid {\text N}(0, \sigma^2), \quad i = 1,\dots,n,
    \end{align} 
    where $f : [0,1]^p \rightarrow \mathbb{R}$ is a smooth function, $x_i = (x_{i1},\dots,x_{ip})^T\in[0,1]^p$ are fixed design points, and $\sigma^2>0$ is a variance parameter.
A classical strategy for estimating $f$ is to represent it using a tensor-product basis expansion, which allows for flexible modeling of smooth functions in multiple dimensions \citep{schumaker2007spline}:
\begin{align}
f(u_1,\dots,u_p)=\sum_{j_1=1}^{J_1}\cdots\sum_{j_p=1}^{J_p} \theta_{j_1\dots j_p} \prod_{m=1}^p\psi_{j_m}(u_{m}),
\label{eqn:tensor}
\end{align}
where $\{\psi_j\}_{j=1}^{J_m}$ are basis functions for the $m$th coordinate and $\theta_{j_1\dots j_p}$ are the coefficients. 
Among the various options available for this construction, we adopt B-spline basis functions for $\{\psi_j\}_{j=1}^{J_m}$ for each $m=1,\dots,p$.
While this approach is both straightforward and powerful, it is important to note that the quality of estimation is highly sensitive to the number and location of knots. 
Accordingly, it is essential to carefully control the regularity of the regression function to avoid both underfitting and overfitting, ensuring that the model appropriately adapts to the inherent smoothness of the unknown target function. 

One common approach to determine a reasonable knot placement relies on Bayesian model selection (BMS) to establish the optimal number of knots while assuming they are equally spaced according to either the Lebesgue or the empirical measure \citep{kang2024model}.
This method does not impose explicit smoothness penalties on the spline coefficients; instead, the smoothness of the function is inherently characterized by the placement of knots, which are determined based on their posterior probabilities inferred from the data. This approach has been extensively studied and appears in the literature under different terminologies \citep{ghosal2000convergence, rivoirard2012posterior, shen2015adaptive}.
For consistency, we refer to it as Bayesian basis exploration (BBE) in this study.

Choosing a suitable prior distribution is crucial for BBE from both theoretical and practical perspectives. Theoretically, a well-structured prior ensures that the posterior contracts at an optimal rate, even when the smoothness of the target function is unknown \citep{de2012adaptive,shen2015adaptive}. Practically, a suitably informative prior is necessary to avoid Bartlett's paradox in BMS \citep{Bartlett,moreno1998intrinsic}. With these considerations in mind, conventional priors are designed to identify the optimal knot configuration by using priors interpreted as a ridge-type penalty on the coefficients \citep{ishwaran2005spike}. The resulting regularization is then determined solely by the chosen knots, without additional control over the smoothness of $f$.
However, we observe that this limitation often leads to underfitting, as BBE---unlike the true model-nested BMS problem---tends to favor simpler models because of the approximation error induced by the finite truncation of infinite-dimensional parameters; see Section~\ref{sec:simunireg}.
A simple modification, such as using a more informative prior to encourage additional knots, risks overfitting when the target function is highly smooth. 
To address this issue, we find that imposing an additional smoothness penalty on the spline function with a given knot configuration---using, for example, Bayesian P-splines \citep{lang2004bayesian}---encourages BBE to favor more complex models while providing additional control over the smoothness of $f$, thereby reducing the risk of underfitting and overfitting. In this context, incorporating additional smoothness regularization into a prior for BBE can be highly effective.

In this study, we propose a new approach that integrates roughness penalization into BBE. Specifically, we construct a prior distribution for BBE by combining Bayesian P-splines with Zellner’s g-prior \citep{zellner1986assessing}, introducing a weight parameter that governs the strength of the penalty. This framework adaptively determines the basis functions while ensuring the desired level of smoothness in the estimated function from both theoretical and practical perspectives. From a theoretical standpoint, we demonstrate that the proposed method attains the optimal posterior contraction rate while adapting to unknown smoothness levels. From a practical standpoint, our numerical study shows that the proposed method consistently outperforms existing alternatives or achieves comparable results with improved model efficiency. Figure~\ref{plot:posteriormeancurves} illustrates the empirical performance of our approach across three bivariate functions, showing that, relative to other Bayesian competitors, it yields the smallest mean squared error (MSE), defined as $\int_{[0,1]^p} (f(x)-\hat f(x))^2 dx$, where $\hat f$ is the pointwise posterior mean. The R implementation of the proposed method is available at \url{https://github.com/Damelim/BPBS}.

\begin{figure}[t!]
    \centering 
\includegraphics[width=6.3in]{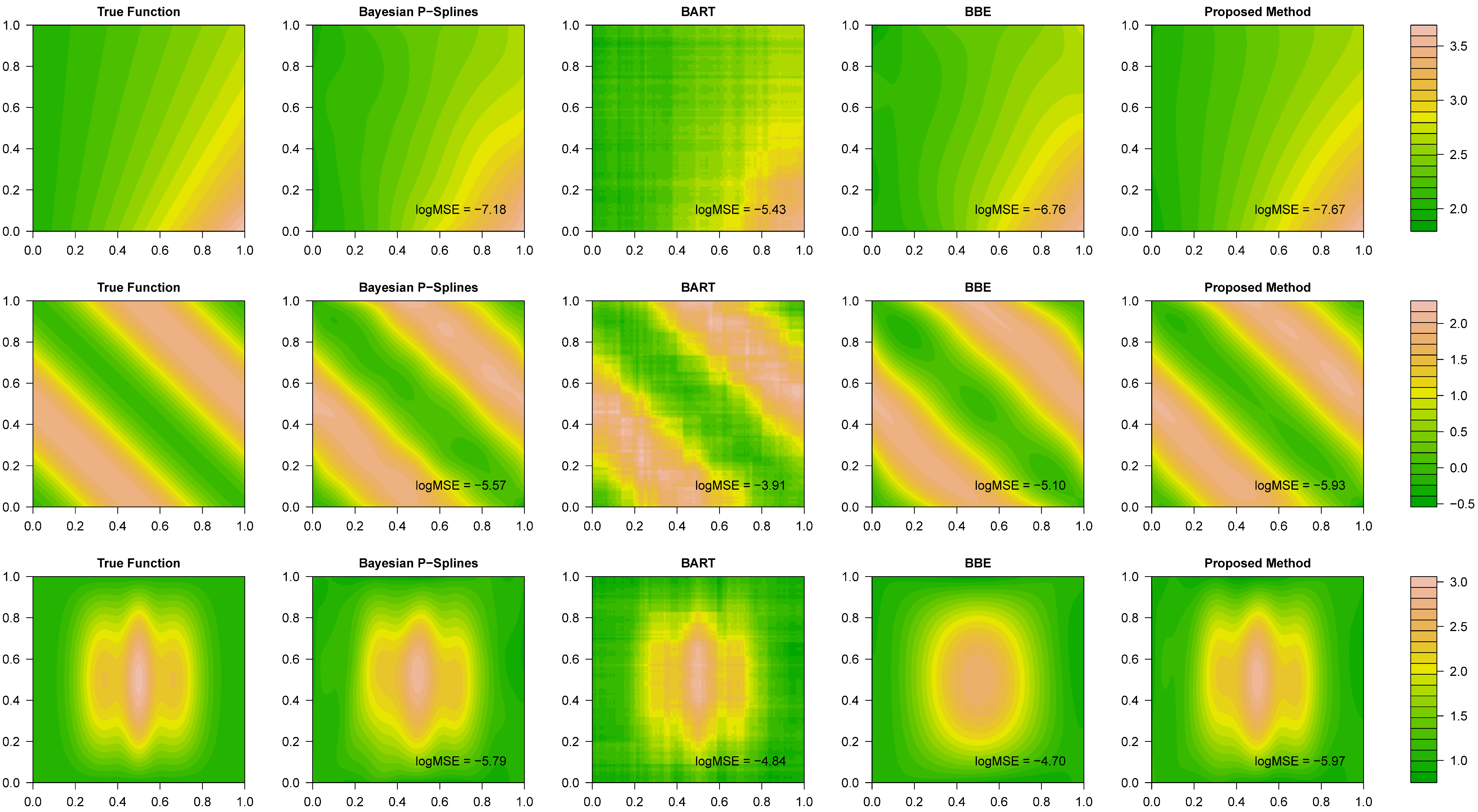}
    \caption{Comparison of the proposed method with Bayesian competing methods: Bayesian P-splines, Bayesian additive regression trees (BART), and BBE (with the Zellner-Siow prior; see Section~\ref{sec:simunireg}). Each dataset is generated from the model in \eqref{eqn:nonparamreg} with each function using $\sigma^2=0.5^2$ and $n=4000$. For each method, the panel displays the contour plot of the pointwise posterior mean and the logarithm of the MSE.}
    \label{plot:posteriormeancurves}
\end{figure}

   The remainder of the paper is organized as follows. 
	Section~\ref{basisselection} provides an overview of BBE and its conventional priors. Section~\ref{proposedmethod} introduces the proposed method, which incorporates roughness penalization into BBE, along with a well-specified prior distribution and a Markov chain Monte Carlo (MCMC) algorithm.
Section~\ref{sec:theory} presents the theoretical analysis of the proposed method, establishing the posterior contraction rate. Section~\ref{sec:simulation} reports the results of an extensive simulation study, and Section~\ref{sec:realdata} demonstrates the application of the proposed method to real datasets. Technical proofs and additional details on the sampling algorithm are provided in the Appendix.
    
  \subsection{Notation}
  
  For a function $g:[0,1]^p\rightarrow \mathbb R$, $\lVert g \rVert_n = (n^{-1}\sum_{i=1}^n (g(x_i))^2)^{1/2}$ denotes the empirical $L_2$-norm relative to the design points $x_i$, $i=1,\dots,n$.
The notation $\rho_{\min}(\cdot)$ denotes the minimum eigenvalue of a square matrix. For sequences $a_n$ and $b_n$, writing $a_n\lesssim b_n$ (or equivalently $b_n\gtrsim a_n$) means that there exists a constant $C$ independent of $n$ such that $a_n\le C b_n$. We write $a_n\asymp b_n$ if both $a_n\lesssim b_n$ and $b_n\lesssim a_n$ hold. We use $\Pi$ to denote the prior or posterior distribution and $\pi$ to denote the corresponding density with respect to an appropriate dominating measure.

    \section{Bayesian Basis Exploration}\label{basisselection}
	Here we review BBE and its limitation, which motivates the proposed method. 
    BBE determines the optimal specification of knots in a data-driven manner using BMS. 
    Conventionally, this idea aligns with the principles of Bayesian model averaging \citep{hoeting1999bayesian}, as it involves averaging over various knot configurations based on the posterior distribution, rather than choosing a single configuration with the highest posterior probability.
    For each $m=1,\dots, p$, we employ B-spline basis functions defined by a set of predetermined knots, denoted as $0 = \xi_{m0} < \xi_{m1} < ... < \xi_{mK_m} < \xi_{m,K_m+1}=1$.    
    The boundary knots can be replaced with $\xi_{m0}=\min_i x_{im}$ and $\xi_{m,K+1}=\max_i x_{im}$ if the domain boundary is unknown.
    As a result, for a spline order $l_m\ge 1$ (polynomial degree $l_m-1$) and $K_m\ge 0$ interior knots, we generate $J_m = l_m + K_m$ terms of B-spline basis functions $\{\psi_j\}_{j=1}^{J_m}$ for coordinate $m$, whereby the function $f:[0,1]^p\rightarrow \mathbb R$ is expressed by \eqref{eqn:tensor}. For notational simplicity, we define $\J=\{J_m\}_{m=1}^p$ and $J=\prod_{m=1}^p J_m$. We restrict ourselves to at least cubic-order B-splines (i.e., $l_m\ge 4$), which ensure that $J\ge 4$ for any $p\ge 1$. 
     
	Choosing an appropriate set of knots is crucial for ensuring the reasonable regularity of the function estimates. To accomplish this, BMS can be employed to explore the model space for knot specifications. Several approaches have been proposed, such as determining the knot locations entirely freely \citep{denison1998automatic,dimatteo2001bayesian} and selecting important knots from a set of potential candidates \citep{smith1996nonparametric,kohn2001nonparametric}. As defined for BBE in Section~\ref{sec:intro}, we explore varying the number of knots while arranging their locations equidistantly for each coordinate \citep{shen2015adaptive}. Alternatively, knots can be positioned based on the sample quantiles of the design points \citep{kang2024model}. In this study, we specifically focus on controlling the number of knots with the equidistant arrangement.

Let $B_\J\in\mathbb R^{n\times J}$ be the design matrix and $\theta_\J\in\mathbb R^J$ be its corresponding coefficient vector induced from the expression in \eqref{eqn:tensor}. Specifically, $B_\J$ can be expressed as $B_\J = B_{J_1}\odot \dots \odot B_{J_m}$, where $\odot$ denotes the row-wise Kronecker product, and $B_{J_m}\in\mathbb R^{n\times J_m}$ is the basis matrix for the univariate marginal in coordinate $m$, with its $(i,j)$ element given by $\psi_j(x_{im})$ \citep{wood2006low}.
 The model
in \eqref{eqn:nonparamreg} can be expressed in matrix form as
\begin{align}
y=B_\J\theta_\J + \varepsilon,\quad \varepsilon\sim \text{N}_n(0_n,\sigma^2 I_n),
\label{eqn:basismodel}
\end{align}
where $y=(y_1,\dots,y_n)^T\in\mathbb R^n$.
	Given the nature of BMS, it is crucial to use an informative prior for $\theta_\J$ \citep{moreno1998intrinsic}. A conventional approach is assigning a conditional normal prior
    \begin{align} 
	\theta_\J\mid\J,\sigma^2,\lambda \sim \text{N}_J(0_J,\lambda\sigma^2 \Psi_\J),
		\label{eqn:normpri}
    \end{align}
	with a positive definite matrix $\Psi_\J$ and a dispersion factor $\lambda$ controlling the strength of prior information. Choosing appropriate $\Psi_\J$ is important because $\lambda$ cannot be arbitrarily large owing to Bartlett's paradox \citep{Bartlett}.	
	Common choices include $\Psi_\J=I_J$, which is reasonable given the similar scale of the columns in $B_\J$. Another prevalent option is using the g-prior with $\Psi_\J=n(B_\J^T B_\J)^{-1}$ \citep{zellner1986assessing}, assuming that $B_\J$ has full column rank. The dispersion parameter $\lambda$ can be treated either as fixed or as a parameter with an assigned prior distribution \citep{liang2008mixtures}.

Although the prior in \eqref{eqn:normpri} is widely used, it imposes strong prior information on the height of the function by shrinking it toward zero, resulting in a lack of invariance to shifts in the target function. To address this issue, one effective approach is to assign a nearly or entirely flat prior to the global mean term $\theta_0$ of $f$, while imposing an informative prior on the remaining signal $\tilde{f} = f - \theta_0$ under an identifiability constraint.
Specifically, we write $f(\cdot) = \theta_0 + \tilde f(\cdot)$ with the constraint $\sum_{i=1}^n \tilde f(x_i) = 0$ \citep{xue2009consistent,jeong2022bayesian}. To complete the reparameterization, define the transformed basis matrix $\tilde B_\J =(I_n-n^{-1}1_n1_n^T)B_{\J,-1}$, which is orthogonal to $1_n$, where $B_{\J,-1}\in\mathbb R^{n\times (J-1)}$ is the matrix obtained by removing the first column of $B_\J$. 
Then, one can show that $B_\J\theta_\J = \theta_0 1_n + \tilde B_\J\tilde \theta_\J$ for some $\tilde \theta_\J\in\mathbb R^{J-1}$ (see the proof of Lemma~\ref{lmm:pspline} below). Consequently, $\tilde f(x_i)$ is expressed as the $i$th entry of $\tilde B_\J\tilde \theta_\J$.

A suitable prior should be assigned to $\tilde \theta_\J$, while a flat prior can be assigned to the global mean term $\theta_0$ to ensure invariance to shifts. Unlike in \eqref{eqn:normpri}, using the identity matrix in the prior covariance is now less advisable owing to the transformation. For example, if the transformed basis matrix is obtained by removing the last column instead of the first, the identity prior covariance induces different prior information. 
A natural choice is the g-prior with centering owing to its invariance property \citep{zellner1986assessing}; that is,
    \begin{align} 
    	\begin{split}
    		\pi(\theta_0) &\propto 1, \\
    		\tilde \theta_\J \mid \J, \sigma^2, \lambda &\sim \text{N}_{J-1}(0_{J-1}, \lambda \sigma^2 n( \tilde B_\J^T \tilde B_\J)^{-1}).
    	\end{split}
    	\label{eqn:gpriortransformed}
    \end{align}
 The improper prior for $\theta_0$ can alternatively be replaced by a normal prior with a large variance.
 It is evident that the prior in \eqref{eqn:gpriortransformed} is invariant to shifts in $y$ (or, equivalently, the target function $f$), and to scaling of $y$ provided that the Jeffreys prior is used for $\sigma^2$. Moreover, it exhibits invariance under any invertible linear transformation of the design matrix, which is a desirable property of the g-prior \citep{liang2008mixtures}.
	Observe that the prior in \eqref{eqn:gpriortransformed} penalizes deviations from the center $\tilde f=f-\theta_0$ through the term $\tilde \theta_\J^T\tilde B_\J^T \tilde B_\J\tilde \theta_\J$ as a ridge-type penalty, which can be approximated by $\int( f(u)-\theta_0)^2 du$ if the design points are sufficiently regular. Therefore, BBE with the g-prior can be viewed as the process of choosing the optimal $\J$ while penalizing the magnitude of $f-\theta_0$. 
	
	With an inverse gamma prior or the Jeffreys prior on $\sigma^2$, the marginal likelihood $p(y \mid \J,\lambda)$ can be analytically derived in closed form for both \eqref{eqn:normpri} and \eqref{eqn:gpriortransformed}. In addition, certain prior distributions for $\lambda$ allow it to be integrated out \citep{liang2008mixtures, maruyama2011fully}, making the fully marginal likelihood $p(y \mid \mathcal{J})$ available. This renders the exploration of the optimal $\J$ using BBE computationally straightforward. Furthermore, with a suitably chosen prior on $\J$ that exhibits exponential decay, BBE benefits from well-established asymptotic properties that allow adaptation to unknown levels of smoothness \citep{de2012adaptive,belitser2014adaptive,shen2015adaptive}, provided that a proper prior is used. Notably, this adaptive property has not yet been achieved in Bayesian penalized splines \citep{bach2021posterior}.
	
However, despite these theoretical and practical advantages, BBE has a limitation in that function smoothness depends solely on the number and location of the basis terms. This can lead to substantial estimation bias stemming from the approximation error between the true function and the spline approximation, particularly because BMS tends to favor parsimonious models, thereby disregarding the infinite-dimensional nature of the true model. More specifically, because the spline-based prior does not encompass the true data-generating process, the estimation procedure tends to underfit rather than overcomplicate models for marginal improvements. This negative feature of BBE is evident in our simulation study in Section~\ref{sec:simunireg}. One might consider a simple modification that employs a more informative prior favoring additional knots, but this approach risks overfitting when the underlying function is sufficiently smooth. Instead, we propose incorporating a roughness penalty term into the BBE prior, which ensures the selection of a sufficient number of knots while simultaneously providing additional smoothness control. The proposed method is detailed in the next section.

    \section{Proposed Method}
    \label{proposedmethod}
    
In this section, we present the proposed method, which incorporates the P-spline penalty into the prior construction for BBE. The resulting prior features two components in the precision matrix: a roughness penalty that regulates the smoothness of $f$, and a ridge-type penalty that controls the magnitude of $f-\theta_0$. We begin by providing an overview of Bayesian P-splines and discussing how they are adapted for our purposes. Next, we detail the construction of the prior and present a computationally efficient MCMC algorithm for obtaining the model-averaged posterior distribution.

    \subsection{Recasting Bayesian P-Splines}
    \label{recastingbpspline}
    
    Bayesian P-splines are viewed as the Bayesian counterpart to penalized splines \citep{lang2004bayesian}, which have achieved widespread success and have found applications in numerous function estimation problems. For multidimensional smoothing with the tensor product of B-splines, Bayesian P-splines can be extended from the univariate construction.
    Specifically, for the model in \eqref{eqn:basismodel}, the penalized spline penalty $\theta_\J^T P_\J \theta_\J$ can be expressed using the  matrix $P_\J = \sum_{m=1}^p P_{J_m}\in\mathbb R^{J\times J}$, where $P_{J_m} = I_{J_1}\otimes I_{J_{m-1}}\otimes D_{J_m}^T D_{J_m} \otimes I_{J_{m+1}}\otimes I_{J_p}$ is the partial penalty matrix, with $D_{J_m}$ representing a matrix representing a finite difference operation \citep{wood2006low,wood2017p}. We use the second-order finite difference matrix defined as
    \begin{align*}
    	D_{J_m} &= 
    	\begin{pmatrix} 
    		1 & -2 & 1 & 0 & \cdots & 0 & 0 & 0 & 0 \\
    		0 & 1 & -2 & 1 & \cdots & 0 & 0 & 0 & 0 \\
    		\vdots & \vdots & \vdots & \vdots & \ddots & \vdots & \vdots & \vdots & \vdots \\
    		0 & 0 &  0 & 0  & \cdots & 0 & 1 & -2 & 1
    	\end{pmatrix} \in \mathbb{R}^{(J_m-2) \times J_m}.
    \end{align*}
 Using this expression, Bayesian P-splines are characterized by putting an improper prior on the spline coefficients $\theta_\J$, 
    \begin{align}
    	\pi(\theta_\J\mid\J,\sigma^2,\lambda) &\propto \frac{|P_\J|_+^{1/2}}{(\lambda\sigma^2)^{\textnormal{rank}( P_\J)/2}}\exp \!\left(-\frac{1}{2\lambda\sigma^2} \theta_\J^T P_\J \theta_\J \right),\quad \theta_\J\in\mathbb R^J,
    	\label{eqn:psplineprior}
    \end{align}
where $|\cdot|_+$ denotes the pseudo-determinant.
Compared to the original construction in \citet{lang2004bayesian}, we parameterize the smoothing parameter as $\lambda\sigma^2$ rather than $\lambda$, ensuring that the prior remains invariant to the scaling of the response vector $y$ while also providing conjugacy when combined with the BBE prior in Section~\ref{priorspecification}. The dispersion parameter $\lambda$ plays a crucial role in determining the level of smoothness and is commonly assigned an inverse gamma prior. Instead of using a common $\lambda$, different smoothing parameters can be assigned to each $P_{J_m}$ to account for anisotropic smoothness; see, for example, \citet{wood2006low}.

    Our objective is to incorporate the penalized spline penalty into the BBE prior to introduce roughness penalization. We are particularly interested in the form of \eqref{eqn:gpriortransformed} for the mean response $\theta_0 1_n + \tilde B_\J\tilde \theta_\J$ owing to its invariance properties, rather than \eqref{eqn:normpri} for $B_\J\theta_\J$. However, there is a discrepancy between the parameterizations in \eqref{eqn:gpriortransformed} and \eqref{eqn:psplineprior}; the former is for the transformed coefficients $(\theta_0,\tilde\theta_\J)$, while the latter is for the B-spline coefficients $\theta_\J$. To align the two expressions, we transform the Bayesian P-spline prior in \eqref{eqn:psplineprior} so that it applied to $(\theta_0,\tilde\theta_\J)$. The result is presented in the following lemma. The proof is provided in Appendix~\ref{appendix_proof}.

    \begin{lemma}[Transformed P-spline prior] \label{lmm:pspline}
    	The Bayesian P-spline prior for $(\theta_0,\tilde\theta_\J)$, derived from the original prior in \eqref{eqn:psplineprior}, is given by
    	\begin{align}
    		\begin{split}
    			\pi(\theta_0) &\propto 1 ,\\
    			\pi(\tilde\theta_\J\mid\J,\sigma^2,\lambda) &\propto \frac{|\tilde P_\J|_+^{1/2}}{(\lambda\sigma^2)^{\textnormal{rank}(\tilde P_\J)/2}}\exp \!\left(-\frac{1}{2\lambda\sigma^2} \tilde\theta_\J^T\tilde P_\J \tilde\theta_\J \right),\quad \tilde\theta_\J\in\mathbb R^{J-1},
    			\label{eqn:psplineprior_modified}
    		\end{split}
    	\end{align}
    	where $\tilde P_\J\in\mathbb R^{(J-1)\times(J-1)}$ is obtained by removing the first row and column of $P_\J$.
    \end{lemma}
	The expression in \eqref{eqn:psplineprior_modified} implies that it exhibits invariance to shifting.    Since both \eqref{eqn:gpriortransformed} and \eqref{eqn:psplineprior_modified} are priors for the transformed coefficients $(\theta_0,\tilde\theta_\J)$, we can combine the two expressions to construct a new prior distribution that capitalizes on the strengths of each approach.
	 Using the definition of the pseudo-determinant, \eqref{eqn:psplineprior_modified} indicates that the P-spline prior for $\tilde\theta_\J$ can be approximated by the density of $\text{N}_{J-1}(0_{J-1},\lambda\sigma^2(\tilde P_\J + \delta I_{J-1})^{-1})$ for a small $\delta>0$, up to a normalizing constant depending on $\J$. The flat prior for $\theta_0$ can be approximated by a normal prior with a large variance. These approximations offer a framework for constructing the combined prior distribution.
 
    \subsection{Construction of the Prior Distribution} 
    \label{priorspecification}

    Observe that both \eqref{eqn:gpriortransformed} and \eqref{eqn:psplineprior_modified} impose a flat improper prior on the global mean parameter $\theta_0$, which can be approximated by a normal prior with large variance.
    For the remaining coefficients associated with the centered signal $\tilde f$, we propose a prior motivated by \eqref{eqn:gpriortransformed} and \eqref{eqn:psplineprior_modified}. The resulting prior distribution on $(\theta_0,\tilde\theta_\J)$ is given by
    \begin{align} 
    	\begin{split}
    		\theta_0\mid \sigma^2  &\sim \text{N}(0, \sigma^2 \kappa^2), \\
    		\tilde\theta_\J\mid \J,\tau,\lambda ,\sigma^2 &\sim \text{N}_{J-1} \!\left(0_{J-1}, \lambda \sigma^2\! \left( (1-\tau) \tilde P_\J + \tau n^{-1}\tilde B_\J^T \tilde B_\J \right)^{-1} \right),
    	\end{split}
    	\label{eqn:Covariance_proposedmethod}
    \end{align} 
    where $\kappa>0$ is a fixed constant, while $\tau\in(0,1)$ and $\lambda>0$ are parameters assigned suitable prior distributions. Observe that the covariance matrix is nonsingular since we assume that $\tilde B_\J^T \tilde B_\J$ is positive definite. The constant $\kappa$ should be chosen sufficiently large to mimic a flat prior for the global mean, ensuring that its estimate remains uninfluenced by other coefficients. Alternatively, one could directly employ the flat prior $\pi(\theta_0) \propto 1$ from a practical perspective, but we use the proper proxy in \eqref{eqn:Covariance_proposedmethod} for theoretical purposes.
    The weight parameter $\tau\in(0,1)$ serves as a means to balance between the two penalty terms, effectively determining the relative importance of model selection versus penalization. 
  If $\tau\rightarrow 1$, the prior in \eqref{eqn:Covariance_proposedmethod} approaches the model selection prior in \eqref{eqn:gpriortransformed}. On the other hand, if $\tau\rightarrow0$, the prior is viewed as a proxy of the P-spline prior in \eqref{eqn:psplineprior_modified}.
   Combined with $\tau$, the dispersion parameter $\lambda$ plays similar roles in BBE and P-splines. More precisely, the factor $\lambda/\tau$ determines deviations from the center through $n^{-1}\tilde B_\J^T \tilde B_\J$, while $\lambda/(1-\tau)$ controls the smoothness of the function through the penalty matrix $\tilde P_\J$.
   
	The prior specification is now completed for the remaining parameters. The entire priors are designed with both theoretical and practical considerations.
   \begin{itemize}
   	\item {\it Prior for $\sigma^2$.} 
  	We assign an inverse gamma prior $\sigma^2\sim \text{IG}(a_\sigma, b_\sigma)$, which is a natural choice owing to conjugacy. The constants $a_\sigma > 0$ and $b_\sigma > 0$ are chosen to be small to closely approximate the Jeffreys prior, which retains invariance properties under scaling. The resulting marginal prior for the coefficients $(\theta_0,\tilde\theta_\J)$ is a $t$-prior, whose polynomial tails pose challenges for theoretical analysis. To address this, we restrict $\sigma^2$  in its marginal posterior distribution with high posterior probability; see Lemma~\ref{lmm:sigma}.
  	   	  	
   	\item {\it Prior for $\tau$.} 
   	As specified in Assumption~\ref{asm:prior} below, the theoretical analysis requires that the prior for $\tau$ exhibit exponential decay toward zero to ensure the positive definiteness of the covariance matrix with high prior probability. Empirical results, however, suggest that performance improves when the prior for $\tau$ is shifted closer to one for small samples and gradually moves toward zero as the sample size increases, imposing a stronger roughness penalty for larger samples. To balance these considerations, we adopt a truncated beta prior $\text{Beta}_{(\delta, 1)}(c_\tau n^{-1}, 1)$ with $c_\tau > 0$ and small $\delta > 0$, where $\text{Beta}_A$ denotes a beta distribution truncated to $A$.
   	
   	\item {\it Prior for $\lambda$.} Similarly, Assumption~\ref{asm:prior} requires the prior for $\lambda$ to exhibit exponential decay while maintaining sufficient mass over positive reals bounded away from zero. An inverse gamma prior, which appears natural owing to its semi-conjugacy, only provides a polynomial tail on the right side. Among the distributions with exponential tails, we employ an exponential prior distribution with rate $c_\lambda>0$, as its monotone density facilitates the construction of an MCMC algorithm using data augmentation; see Section~\ref{MCMC_algorithm} for more details.
   	
   	\item  {\it Prior for $\J$.} Assumption~\ref{asm:prior} requires a Poisson-type tail for $\mathcal J$. Instead of a Poisson prior, we choose a density that is monotonically decreasing, which is intuitively desirable for penalizing more complex models. Specifically,
   	\begin{align}
   	\pi(\J) \propto \left(\frac{\nu}{J-\prod_{m=1}^p l_m+\nu/e}\right)^{J-\prod_{m=1}^p l_m+\nu/e}\mathbbm 1\{\rho_{\min}(B_\J^T B_\J)>0\},
   	\label{eqn:priorj}
   	\end{align}
   	where $\nu>0$ is a fixed constant. The factor $\nu/ e$ makes the density monotonically decreasing. Combined with Assumption~\ref{asm:eigen} in Section~\ref{sec:theory}, it satisfies the required tail properties. 
   	Given that $\nu$ encapsulates prior information about $J$, it is reasonable to set $\nu$ to grow exponentially with $p$, that is, $\log \nu \propto p$.
   \end{itemize}

\subsection{Efficient Monte Carlo Sampling Algorithm} \label{MCMC_algorithm}

We present an efficient MCMC algorithm for sampling from the joint posterior distribution $\pi(\J, \sigma^2, \theta_0, \tilde\theta_\J, \tau, \lambda \mid y)$. Our carefully designed prior distribution enables sampling from the posterior distribution to be straightforward and easily implementable. Once the MCMC draws are collected, one can produce the model-averaged posterior of any functional of $f$, for example, pointwise evaluations at specific values or credible bands of the function.
The following is a blocked Gibbs sampler that alternates between the full conditionals of $(\J, \sigma^2, \theta_0, \tilde\theta_\J)$, $\tau$, and $\lambda$.

\begin{enumerate}[label=(\roman*)]
    \item Draw $J_m$ from $\pi(J_m \mid \J_{-m}, \tau, \lambda, y)\propto \pi(\J)p(y \mid  \J, \tau, \lambda)$ for $m=1,\dots,p$, where $\J_{-m}$ is the set of $\{J_k,k\ne m\}$ and $p(y \mid  \J, \tau, \lambda)$ is the marginal likelihood given by 
    \begin{align*}
    p(y \mid  \J, \tau, \lambda)
    &\propto 
    \big| I_{J-1} - \Omega_{\J,\tau,\lambda}^{-1} \tilde B_\J^T \tilde B_\J
    \big|^{1/2} \\
    &\quad\times
    \left(
    b_\sigma + 
        \frac{1}{2} \left[y^T y - \frac{(y^T 1_n)^2}{n + \kappa^2} -
        y^T \tilde B_\J
        \Omega_{\J,\tau,\lambda}^{-1}
       \tilde B_\J^T y
        \right] \right)^{-(a_\sigma + n/2)},
    \end{align*}
with $\Omega_{\J,\tau,\lambda} = 
           (1-\tau)\lambda^{-1}  \tilde P_\J +
           (n + {\tau}/{\lambda} ) n^{-1} \tilde B_\J^T \tilde B_\J$.
We draw $J_m$ using the random walk Metropolis-Hastings algorithm by proposing a value based on its current value.

\item Draw $\sigma^2$ from $\pi(\sigma^2 \mid \J, \tau, \lambda,y)$, where 
$$\sigma^2  \mid  \J, \tau, \lambda, y \sim
\text{IG} \!\left( a_\sigma+\frac{n}{2} ,  
            b_\sigma+
            \frac{1}{2}
            \left[y^T y - \frac{(y^T 1_n)^2}{n + \kappa^2} -
            y^T \tilde B_\J
            \Omega_{\J,\tau,\lambda}^{-1}
            \tilde B_\J^T y
            \right] \right)
        .$$
\item Draw $(\theta_0,\tilde\theta_\J)$ from $\pi(\theta_0,\tilde\theta_\J \mid  \J, \sigma^2, \tau, \lambda, y)$, where
\begin{align*}
          \theta_0  \mid  \J, \sigma^2, \tau, \lambda, y
          &\sim 
          \text{N} \!\left( y^T 1_n / (n + \kappa^{-2}) , \sigma^2 / (n + \kappa^{-2})\right),\\
          \tilde\theta_\J  \mid  \J, \sigma^2, \tau, \lambda, y
          &\sim 
          \text{N}_{J-1} 
          \!\left( 
          \Omega_{\J,\tau,\lambda}^{-1}
          \tilde B_\J^T y, 
          \sigma^2 \Omega_{\J,\tau,\lambda}^{-1} \right).
      \end{align*}
\item Draw $\tau$ from $\pi(\tau  \mid  \J, \sigma^2, \theta_0,\tilde\theta_\J, \lambda,y) \propto \pi(\tau) \pi(\tilde\theta_\J  \mid  \J, \sigma^2, \tau, \lambda)$ using grid sampling. Specifically, we obtain
\begin{align} \label{gridsampling}
	\begin{split}
		\log\pi(\tilde\theta_\J  \mid  \J, \sigma^2, \tau, \lambda)&=\frac{1}{2} \sum_{k=1}^{J-1} \log \!\left(1 +  \frac{(1-\tau)n}{\tau} \rho_k \!\left( (\tilde B_\J^T \tilde B_\J)^{-1} \tilde P_\J \right) \right) \\
		&\quad + 
		\frac{J-1}{2} \log \tau  - 
		\frac{1-\tau}{2 \lambda \sigma^2} \tilde\theta_\J^T \tilde P_\J \tilde\theta_\J -
		\frac{\tau}{2 n \lambda \sigma^2} \tilde\theta_\J^T \tilde B_\J^T \tilde B_\J \tilde\theta_\J + c,
	\end{split}
\end{align}
where $\rho_k(\cdot)$ denotes the $k$th eigenvalue of a matrix in a decreasing order and $c$ is a constant independent of $\tau$. The sampling is highly efficient since it needs matrix operations and eigenvalue computations only once during the evaluation of the log density at grid points of $\tau$.
The details are provided in Appendix~\ref{appendix_algorithm}.
\item Draw $\gamma$ and $\lambda$ conditionally from $\pi(\gamma, \lambda  \mid  \J, \sigma^2, \theta_0,\tilde\theta_\J, \tau, y) $ using their full conditional distributions, where $\gamma$ is an auxiliary variable constructed via data augmentation. Specifically,
\begin{align*}
    \gamma\mid \lambda , \J, \sigma^2, \theta_0,\tilde\theta_\J, \tau,  y &\sim \text{Unif}(0,h(\lambda;c_\lambda)),\\
\lambda  \mid  \gamma , \J, \sigma^2, \theta_0,\tilde\theta_\J, \tau,  y & \sim \text{IG}_{(0,h^{-1}(\gamma;c_\lambda))}\!\left (\frac{J-3}{2}, \frac{1}{2 \sigma^2}\tilde\theta_\J^T \left((1-\tau) \tilde P_\J + \tau n^{-1} \tilde B_\J^T \tilde B_\J \right) \tilde\theta_\J \right),
\end{align*}
where $h(\cdot;c_\lambda)$ denotes the exponential density with a rate parameter $c_\lambda$ and $\text{IG}_A $ represents an inverse gamma distribution truncated to $A$. Observe that $h^{-1}$ is available in a closed form. This sampling scheme is valid since $J\ge 4$. The details of the derivation are provided in Appendix~\ref{appendix_algorithm}.
\end{enumerate}

\section{Posterior Contraction Rates} \label{sec:theory}

In this section, we investigate a theoretical aspect of the proposed method. Specifically, we establish that our proposed method for multivariate nonparametric regression achieves the optimal posterior contraction rate in the minimax sense up to a logarithmic factor. Coupled with a given semi-metric for the parameter, a posterior contraction rate characterizes the speed at which the posterior distribution contracts towards the true parameter of the data distribution. We show that the proposed method is a fully adaptive procedure within the Bayesian framework, as it achieves the nearly optimal rate without requiring prior knowledge about the smoothness parameter of the true function. In this section only, the true parameters used for data generation are denoted by $f_0$ and $\sigma_0^2>0$.
Below, we outline the conditions required to achieve the optimal posterior contraction.

\begin{enumerate}[label=\rm(A\arabic*)]
\item \label{asm:holder} The true function $f_0$ belongs to the anisotropic H\"older space defined as
$$
\mathcal H^\alpha([0,1]^p)=\left\{f:[0,1]^p\rightarrow \mathbb R; \ \sup_{u\in[0,1]^p}\left |\frac{\partial^{\sum_{m=1}^p r_m}f(u)}{\partial u_1^{r_1}\dots \partial u_p^{r_p}}\right|<\infty, \ 0\le r_m\le \alpha_m, \ m=1,\dots,p\right\},
$$
where $\alpha=(\alpha_1,\dots,\alpha_p)^T\in\mathbb N^p$ is a smoothness parameter.

 \item \label{asm:spline} The order of B-splines satisfies $l_m\ge \alpha_m$, $m=1,\dots,p$.
\item \label{asm:eigen} For every $\J$ such that $J\lesssim (n/\log n)^{p/(2\bar\alpha+p)} $, the design points satisfy $\log\rho_{\min}(B_\J^T B_\J)\gtrsim -\log n$, where $\bar\alpha^{-1}=p^{-1}\sum_{m=1}^p \alpha_m^{-1}$.
\item \label{asm:prior} The coefficients $(\theta_0,\tilde\theta_\J)$ follow the prior in \eqref{eqn:Covariance_proposedmethod}, and the variance $\sigma^2$ is assigned an inverse gamma prior.
The priors for $\tau$ and $\lambda$ satisfy $\log \Pi\{\tau<n^{-k}\}\lesssim -n$, $\log \Pi\{\lambda>\lambda_0\}\gtrsim -\log n$, and $\log\Pi\{\lambda>n^k\}\lesssim - n$  for some  $\lambda_0>0$ and $k>0$. The prior for $\mathcal J$ satisfies $\log\pi(\J)\asymp -J\log J$ for every $\J$ such that $J\gtrsim (n/\log n)^{p/(2\bar\alpha+p)}$.

\end{enumerate}
Assumption~\ref{asm:holder} is conventional in the literature on multivariate nonparametric regression. It is well known that for any $f_0\in\mathcal H^\alpha([0,1]^p)$, the tensor product of B-splines yields the optimal approximation error $\sum_{m=1}^p J_m^{-\alpha_m}$ with respect to the supremum norm, provided that the spline order satisfies $l_m\ge \alpha_m$ as required in Assumption~\ref{asm:spline} \citep[Chapter 12,][]{schumaker2007spline}.
Assumption~\ref{asm:eigen} imposes a minimal condition on the regularity of the design points. Specifically, if the design points form a regular grid on $[0,1]^p$, then we obtain $\rho_{\min}(B_\J^T B_\J)\gtrsim n/J$ for every $\J$ such that $J\lesssim (n/\log n)^{p/(2\bar\alpha+p)} $ (see Remark~2.1 and Lemma~A.9 of \citet{yoo2016supremum}). Our condition in Assumption~\ref{asm:eigen} is significantly weaker than this. Assumption~\ref{asm:prior} outlines the required conditions for prior specification. In particular, the bounds $\log \Pi\{\tau<n^{-k}\}\lesssim -n$ and $\log\Pi\{\lambda>n^k\}\lesssim - n$ ensure exponential tail properties, while  $\log\pi(\J)\asymp -J\log J$ implies a Poisson-type tail behavior. The condition $\log \Pi\{\lambda>\lambda_0\}\gtrsim -\log n$ is very mild and is satisfied by any prior on $(0,\infty)$ that is independent of $n$. The priors described in Section~\ref{priorspecification} clearly satisfy these conditions.

We now study the contraction rate of the proposed method for multivariate nonparametric regression, using a suitably defined semi-metric between $(f,\sigma)$ and the true values $(f_0,\sigma_0)$. The following theorem formally characterizes the posterior contraction rate of the proposed method. It is well known that the minimax rate of convergence for function estimation over $\mathcal H^\alpha([0,1]^p)$ is $n^{-\bar\alpha/(2\bar\alpha+p)}$. The theorem demonstrates that the proposed method achieves a posterior contraction rate that is minimax-optimal up to a logarithmic factor, while adapting to the unknown smoothness parameter $\alpha$ through the prior. Consistent with the BBE literature, this adaptation is achieved by estimating the number of basis terms  \citep{de2012adaptive,belitser2014adaptive,shen2015adaptive}.  The theorem holds uniformly over $f_0\in\mathcal H^\alpha([0,1]^p)$ with $l_m\ge \alpha_m$, $m=1,\dots,p$.

\begin{theorem}[Posterior contraction]
\label{thm:rate}
Suppose that \ref{asm:holder}--\ref{asm:prior} hold.
For every $M_n\rightarrow \infty$, the posterior distribution satisfies
$$
\mathbb E_0\Pi\!\left\{\lVert f-f_0 \rVert_n+|\sigma-\sigma_0|>M_n ((\log n)/n)^{\bar\alpha/(2\bar\alpha+p)}\mid  y\right\}\rightarrow 0,
$$
where $\mathbb E_0$ is the expectation under the true measure with $f_0$ and $\sigma_0^2$.
\end{theorem}

The complete proof is provided in Appendix \ref{sec:thmproof}. Here, we briefly outline our approach. We utilize the testing-based method for posterior contraction theory \citep{ghosal2000convergence,ghosal2007convergence}, which requires a test function whose errors decay exponentially with respect to the chosen semi-metric quantifying the closeness of the parameters. Specifically, we combine an entropy bound with a test function developed by \citet{jeong2025l2norm} to construct a global test function over a suitably chosen sieve; see Lemma~\ref{lmm:globaltest}. A key technical challenge arises from the polynomial tails of the marginal $t$-prior for the coefficients, which are induced by an inverse gamma prior for $\sigma^2$. Heavy-tailed priors pose challenges in achieving Bayesian adaptation to unknown smoothness due to the insufficient decay of the prior probability assigned to a sieve \citep{shen2015adaptive}. Although one potential remedy involves using fractional posterior distributions \citep{jeong2021posterior,agapiou2024heavy}, the usual posterior distributions are typically of primary interest. By directly examining the marginal posterior of $\sigma^2$ to restrict its support, we show that the usual posterior with an inverse gamma prior for $\sigma^2$ achieves the optimal posterior contraction rate; see Lemma~\ref{lmm:sigma}.

\begin{remark}[Extension to non-hypercube domains]
	\label{rmk:domain}
In this study, we focus on the case where $f$ is defined on the hypercube $[0,1]^p$, as this restriction simplifies the presentation of the main ideas and results. If $f$ is instead defined on a smaller domain $\mathcal X\subsetneq [0,1]^p$, the prior in \eqref{eqn:gpriortransformed} might not be valid because $\tilde B_\J$ could be rank deficient. However, even in this situation, our framework can be readily extended to general domains. This extension is achieved by removing unnecessary basis terms from  \eqref{eqn:tensor}---which is equivalent to eliminating zero columns from $B_\J$---and then adjusting $D_{J_m}$ accordingly when constructing $P_{J_m}$. Assumption~\ref{asm:eigen} then only needs to be satisfied with the refined basis matrix, and the posterior contraction in Theorem~\ref{thm:rate} continues to hold. We employ this extension in analyzing the rainfall data in Section~\ref{sec:realdata}.
\end{remark}

\section{Simulation Study} 
\label{sec:simulation}

\subsection{Univariate Regression}
\label{sec:simunireg}
We first conduct a simulation study focusing on regression with univariate nonparametric functions. This section demonstrates that the proposed method effectively addresses the bias issue in BBE for infinite-dimensional functions caused by underfitting. For the data-generating nonparametric function $f$ in \eqref{eqn:nonparamreg}, we use the following trigonometric test functions mapping from $[0,1]$ to $\mathbb R$: $f_1(u)=1 + \sin(2\pi u)$ and $f_2(u)=1 + \sin(2\pi u)+\sin(10\pi u)/4$, along with their $L_2$-projections, $f_1^\ast$ and $f_2^\ast$, onto the cubic spline space with suitably chosen numbers of equally spaced knots. Specifically, $f_1^\ast$ approximates $f_1$ with two knots, while $f_2^\ast$ approximates $f_2$ with 15 knots.
Therefore, the spline approximations $f_j^\ast$ are included in the candidate models for BMS, whereas the transcendental functions $f_j$ require infinitely many piecewise polynomial basis terms for representation. The functions are visualized in Figure~\ref{plot:function}. Although the differences between $f_j$ and $f_j^\ast$ appear marginal in the figure, we will see that the estimation results for these two types of functions differ substantially.
For each test function, we consider six simulation scenarios with combinations of $\sigma^2 \in \{0.1^2, 0.5^2\}$ and $n \in \{200,2000,20000\}$. The design points $x_i$ are independently generated from the uniform distribution on $(0,1)$.

 \begin{figure}[t!]
	\centering
	\includegraphics[width = 5in]{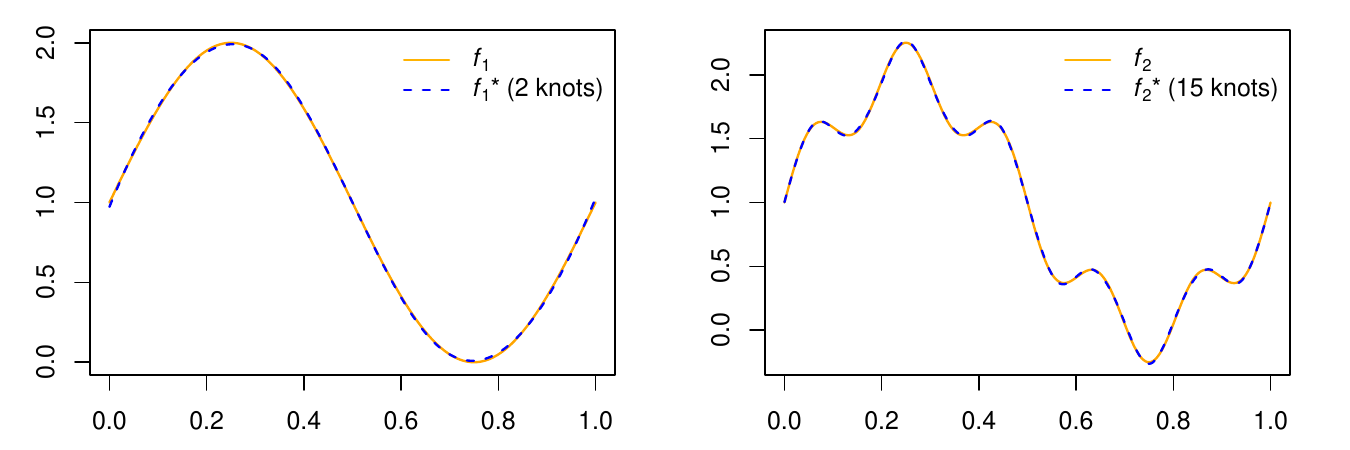}
	\caption{Trigonometric test functions and their cubic spline approximations with equally spaced knots.}
	\label{plot:function}
\end{figure}

We compare the proposed method to BBE with the g-prior defined in \eqref{eqn:gpriortransformed}. For the dispersion parameter $\lambda$, we assign an inverse gamma prior $\text{IG}(1/2,1/2)$, known as the Zellner-Siow prior \citep{zellner1980posterior}, which provides a closed-form expression of the marginal likelihood $p(y \mid \J)$. Although other mixtures of g-priors, such as the beta-prime prior \citep{maruyama2011fully} and the robust prior \citep{bayarri2012criteria}, were also examined, no significant differences in performance were observed. Thus, we include only BBE with the Zellner-Siow prior as a competitor to the proposed method. 
Because this subsection focuses exclusively on how the proposed method addresses the issues of BBE, other function estimation methods are not considered here but will be included in Section~\ref{sec:simmultreg} for multivariate regression.

For the proposed method, we use the prior distributions specified in Section~\ref{priorspecification} with suitably chosen hyperparameters. Specifically, we set $c_\tau=10^4$, which drives the prior for $\tau$ toward zero at a reasonable rate based on our observations. For the exponential prior on $\tau$, we set $c_\lambda=0.32$, which minimizes the Hellinger distance between the exponential distribution and $\text{IG}(1/2,1/2)$. 
This choice is motivated by minimizing the difference in the influence of the priors for $\lambda$ between the proposed method and BBE with the Zellner-Siow prior. For the prior on $\mathcal J$, we set $\nu=100$, which places most of the mass on $J=J_1\le30$ and ensures a reasonable decay in the tails. To allow for a fair comparison, the same prior with $\nu=100$ is also assigned to $\mathcal J$ in BBE. We also tested other values for $c_\tau$, $c_\lambda$, and $\nu$ within reasonable ranges but observed no significant differences in the results.
For each method within every simulation scenario, we generate 500 replications of datasets.  Based on these replications, we evaluate the MSEs and the coverage probabilities of the 95\% pointwise credible bands of BBE and the proposed method. We also calculate the marginal posterior means of $J$ and $\tau$ to understand the performance differences between the two methods.

 \begin{figure}[p!]
	\centering
	\begin{subfigure}[b]{3.1in}
	\includegraphics[width=\linewidth]{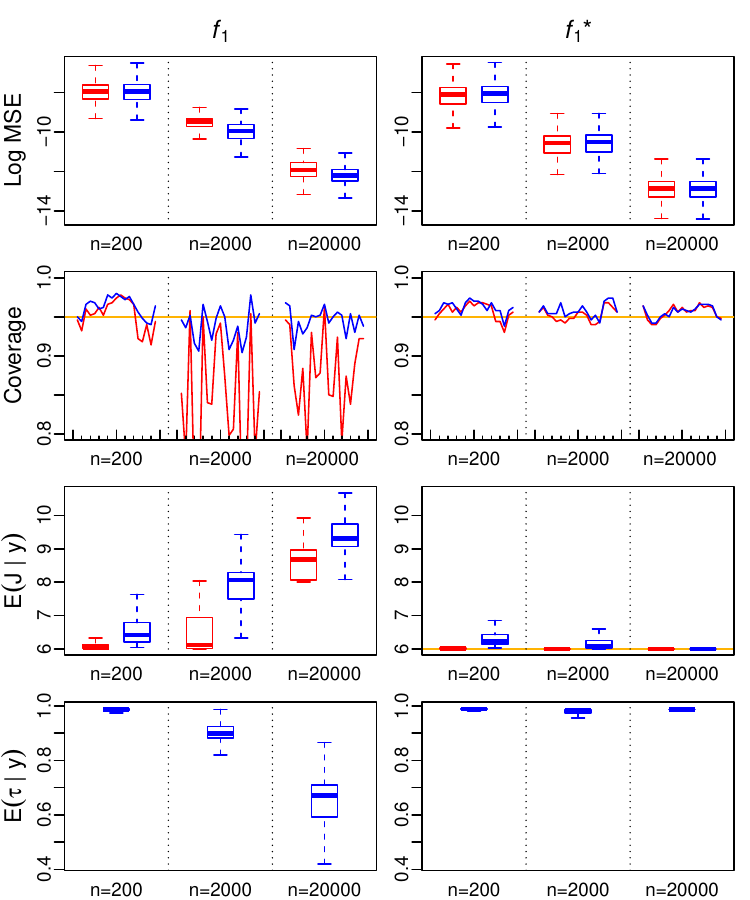}
	\caption{Performance for $f_1$ and $f_1^\ast$ with $\sigma^2=0.1^2$.}
	\label{plot:sim.univ.a}
	\end{subfigure} 
	\begin{subfigure}[b]{3.1in}
	\includegraphics[width=\linewidth]{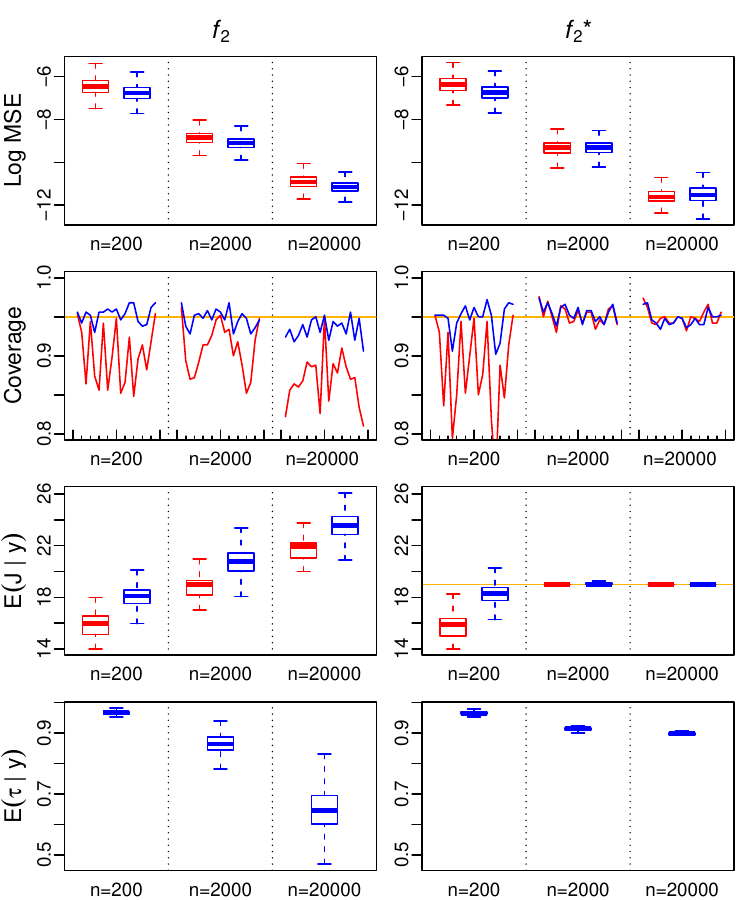}
	\caption{Performance for $f_2$ and $f_2^\ast$ with $\sigma^2=0.1^2$.}
	\label{plot:sim.univ.b}
\end{subfigure} 
	\begin{subfigure}[b]{3.1in}
	\includegraphics[width=\linewidth]{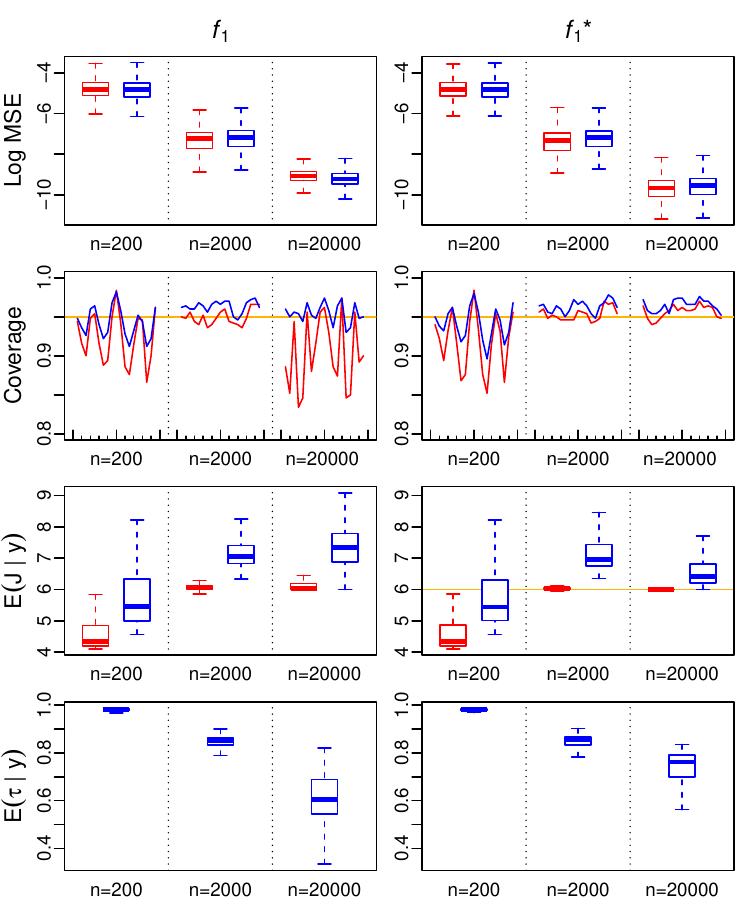}
	\caption{Performance for $f_1$ and $f_1^\ast$ with $\sigma^2=0.5^2$.}
	\label{plot:sim.univ.c}	
\end{subfigure} 
	\begin{subfigure}[b]{3.1in}
	\includegraphics[width=\linewidth]{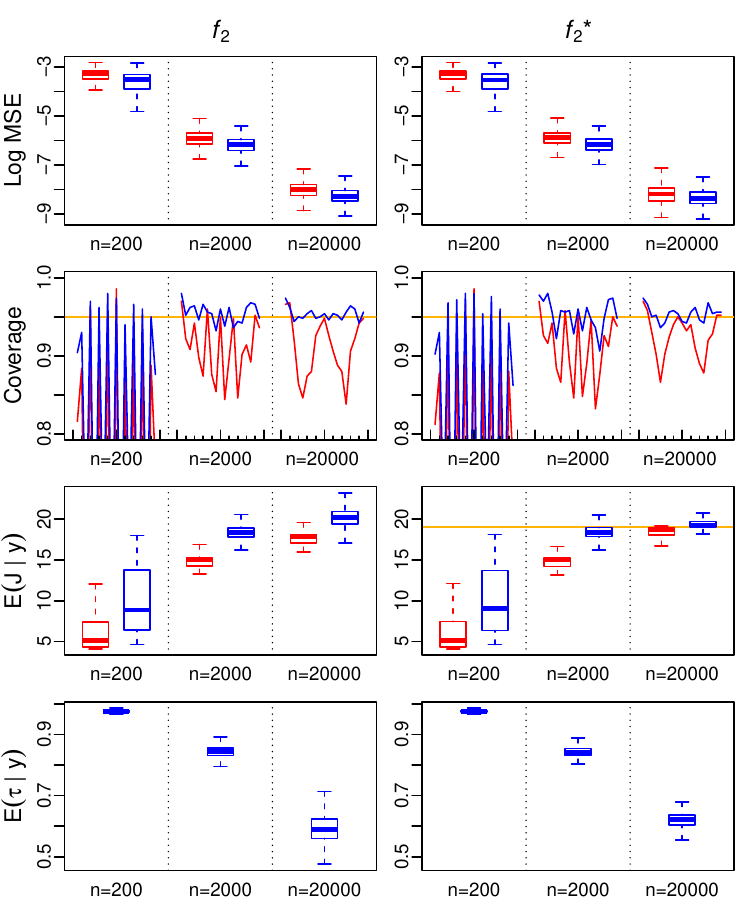}
	\caption{Performance for $f_2$ and $f_2^\ast$ with $\sigma^2=0.5^2$.}
	\label{plot:sim.univ.d}
\end{subfigure} 
	\caption{Performance measures for the univariate regression obtained based on 500 replications: BBE in red ({\color{red}$\blacksquare$}) and the proposed method in blue ({\color{blue}$\blacksquare$}). Each row shows the logarithm of the MSE, the coverage probabilities of the 95\% pointwise credible band on $(0,1)$, the posterior mean of $J$, and the posterior mean of $\tau$. Gold horizontal lines denote the reference values: 95\% coverage and the true $J$ for $f_1^\ast$ and $f_2^\ast$.}
	\label{plot:sim.univ}
\end{figure}

Figure~\ref{plot:sim.univ} summarizes the simulation results. Although Figure~\ref{plot:function} shows only marginal differences between $f_j$ and $f_j^\ast$, the estimation results in Figure~\ref{plot:sim.univ} differ significantly across the two types of functions. This discrepancy arises because the expression for $f_j^\ast$ is included among the candidates for BMS in BBE, leading to nested model selection---a framework where the Zellner–Siow prior has been shown to be effective \citep{liang2008mixtures}. In the higher signal-to-noise ratio (SNR) scenario with $\sigma^2=0.1$, BBE performs reasonably well for $f_1^\ast$ and $f_2^\ast$ but significantly underperforms for $f_1$ and $f_2$, particularly with respect to coverage probability. 
As shown in the panel for the posterior mean of $J$, this performance gap occurs because BBE correctly specifies the values of $J$ for $f_1^\ast$ and $f_2^\ast$, identifying the true models among the candidates. In contrast, BBE produces significant biases for $f_1$ and $f_2$ owing to the inability to precisely identify infinite-dimensional functions.
In comparison, the proposed method performs well for both the trigonometric functions $f_j$ and their spline approximations $f_j^\ast$. For $f_1$ and $f_2$, it addresses the bias issue of BBE by employing more basis terms while avoiding overfitting through additional roughness regularization. For $f_1^\ast$ and $f_2^\ast$, the proposed method behaves similarly to BBE when BBE is effective (by inferring $\tau$ tending toward 1) and outperforms BBE when it underperforms---for example, in the case of $f_2^\ast$ with $n=200$. In the lower SNR scenario with $\sigma^2=0.5^2$, BBE is even less effective than the higher SNR scenario, particularly for $f_2$ and $f_2^\ast$, which exhibit more complex signal variations. In contrast, the proposed method consistently achieves superior performance across all cases.

\subsection{Multivariate Regression}
\label{sec:simmultreg}

Next, we carry out a simulation study involving multivariate nonparametric functions with $p>1$. Unlike Section~\ref{sec:simunireg}, here we compare the proposed method with other Bayesian competitors by assessing performance in terms of MSEs.
We consider the following bivariate and trivariate functions defined on $[0,1]^p$:
\begin{align*}
	f_3(u_1,u_2) &= 1 + \exp({u_1}/({u_2+1})), \\
	f_4(u_1,u_2) &= 1 + \sin(2\pi (u_1+u_2) -{\pi}/{2}), \\
	f_5(u_1,u_2) &= 1 + \frac{1}{3} \!\left(\phi_{1/3,1/10}(u_1) + \phi_{2/3,1/10}(u_1)  + \frac{1}{2}\phi_{1/2,1/20}(u_1)\right)\sin(\pi u_2), \\
	f_6(u_1,u_2,u_3) &= 1 + \sqrt{u_1^2 + (u_2 u_3 - 1/(1+u_2 u_3))^2},\\
	f_7(u_1,u_2,u_3) &= 1 + \sin(2\pi(u_1+u_2)) \cos(2\pi(u_2+u_3)),
\end{align*}
where $\phi_{\mu,\sigma}$ denotes the density of a normal distribution with mean $\mu$ and standard deviation $\sigma$. The bivariate functions $f_3$, $f_4$, and $f_5$ are visualized in Figure~\ref{plot:posteriormeancurves} (from top to bottom). 
Similar to the univariate cases, we consider six simulation scenarios with combinations of $\sigma^2 \in \{0.1^2, 0.5^2\}$ and $n \in \{400,4000,40000\}$, and the design points $x_i$ are independently generated from the uniform distribution on $(0,1)^p$. 

For the Bayesian competitors, we include an anisotropic version of Bayesian P‑splines \citep{bach2025anisotropic}, BART \citep{chipman2010bart}, and BBE with the Zellner–Siow prior. Each method adapts to anisotropic smoothness, ensuring a fair comparison with the proposed approach. The prior specification for our method follows that in Section~\ref{sec:simunireg}, except that we set $\nu=100^p$ to account for the multivariate nature of the functions in the prior for $\mathcal{J}$. For each simulation scenario, we generate 25 replications of the datasets, and we evaluate the MSEs of the methodologies based on these replications.

\begin{figure}[t!]
	\centering
	\begin{subfigure}[b]{2.9in}
		\includegraphics[width=\linewidth]{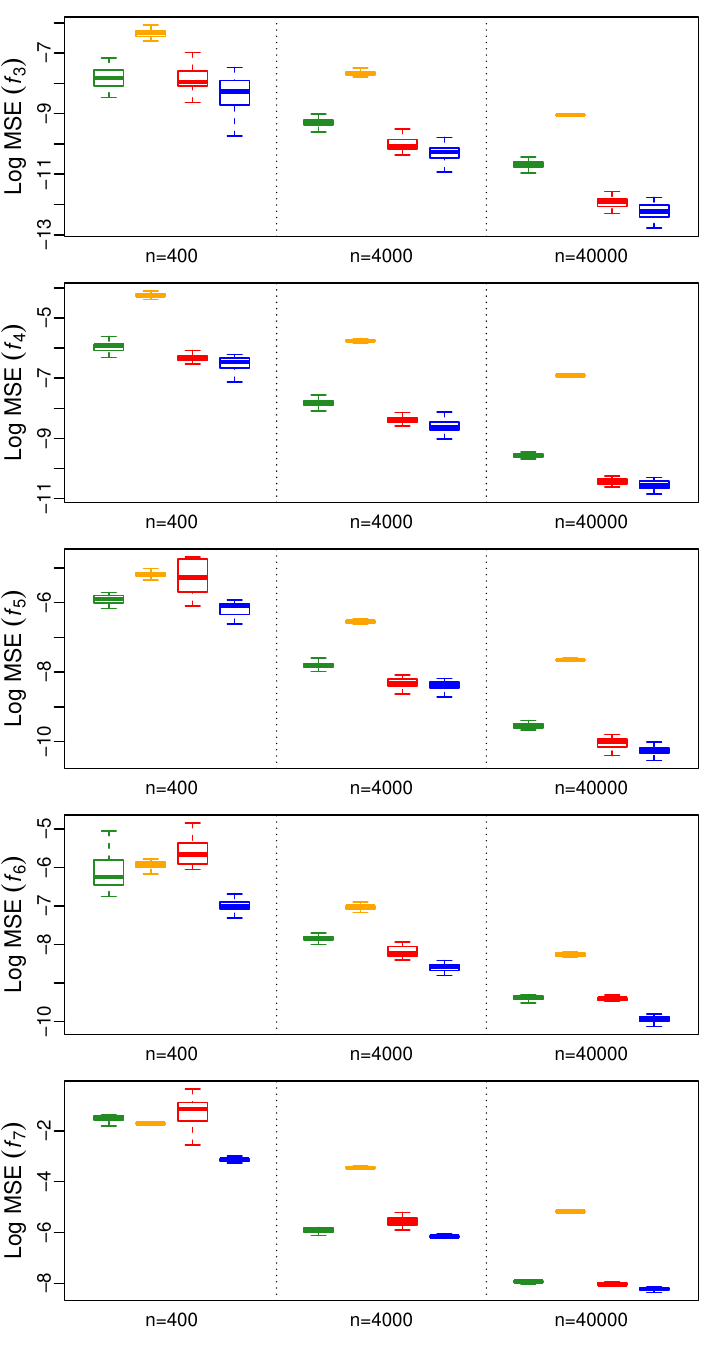}
		\caption{Log of MSEs for $f_3$--$f_7$ with $\sigma^2=0.1^2$.}
		\label{plot:sim.mult.a}
	\end{subfigure} 
	\begin{subfigure}[b]{2.9in}
		\includegraphics[width=\linewidth]{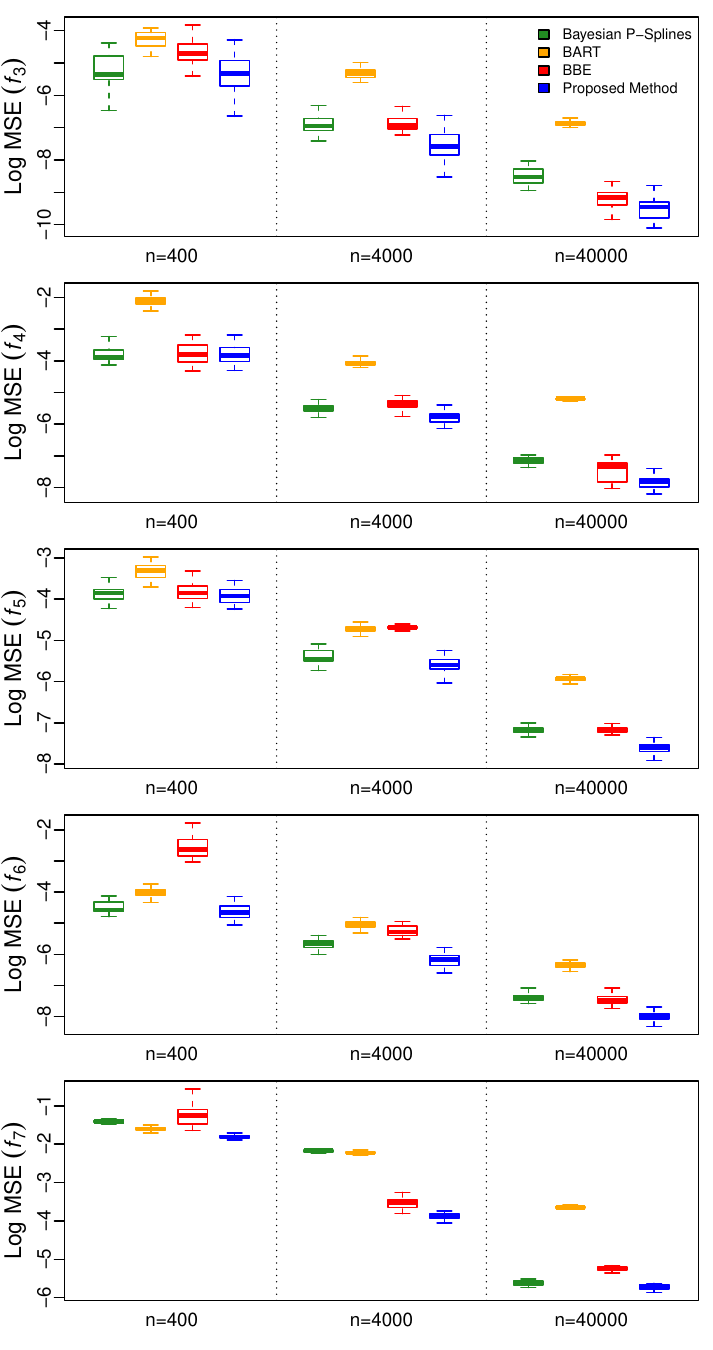}
		\caption{Log of MSEs for $f_3$--$f_7$ with $\sigma^2=0.5^2$.}
		\label{plot:sim.mult.b}
	\end{subfigure} 
	\caption{Logarithm of the MSEs for the multivariate regression obtained based on 25 replications.}
	\label{plot:sim.mult}
\end{figure}

The results are summarized in Figure~\ref{plot:sim.mult}. It is difficult to identify a clear winner between Bayesian P-splines and BBE, as their relative performance varies with the simulation setup. In contrast, BART generally underperforms compared to the other methods, likely because it is a piecewise learner with splits occurring only along coordinate axes. Nevertheless, its performance is expected to improve in higher-dimensional settings \citep{jeong2023art}. Overall, the proposed method consistently achieves the lowest MSE among the competitors, underscoring its effectiveness in smooth function estimation for multivariate regression.

\section{Application to Real Datasets}
\label{sec:realdata}

In this section, we apply the proposed method to two real datasets: the rainfall dataset and the CO2 concentration dataset. The rainfall dataset provides average rainfall measurements (in tenths of millimeters) for June, July, and August across North America. For our analysis, we focused on a subset of $n = 1197$ observations from the United States, with rainfall as the response variable and longitude and latitude as predictors. The CO2 concentration dataset offers synthetic CO2 measurements (in ppm) from $n = 26633$ locations worldwide. For regression analysis, CO2 concentration is used as the response, while longitude and latitude are used as predictors. Both datasets are available in the R package \texttt{fields}. The observations are visualized in Figure~\ref{fig:data}.

   \begin{figure}[t!]
	\centering
	\includegraphics[width = 2.8in]{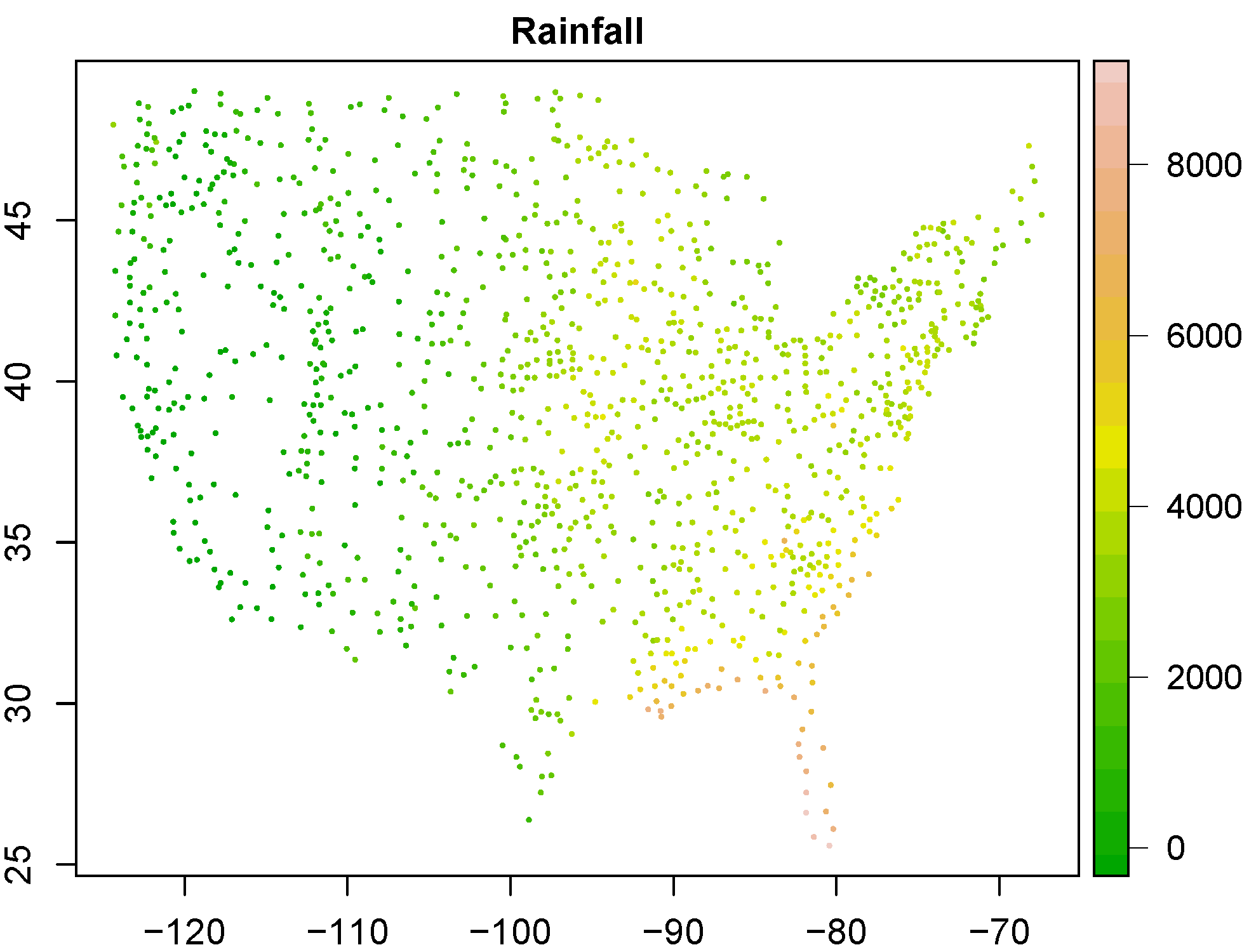} ~ 
	\includegraphics[width = 2.8in]{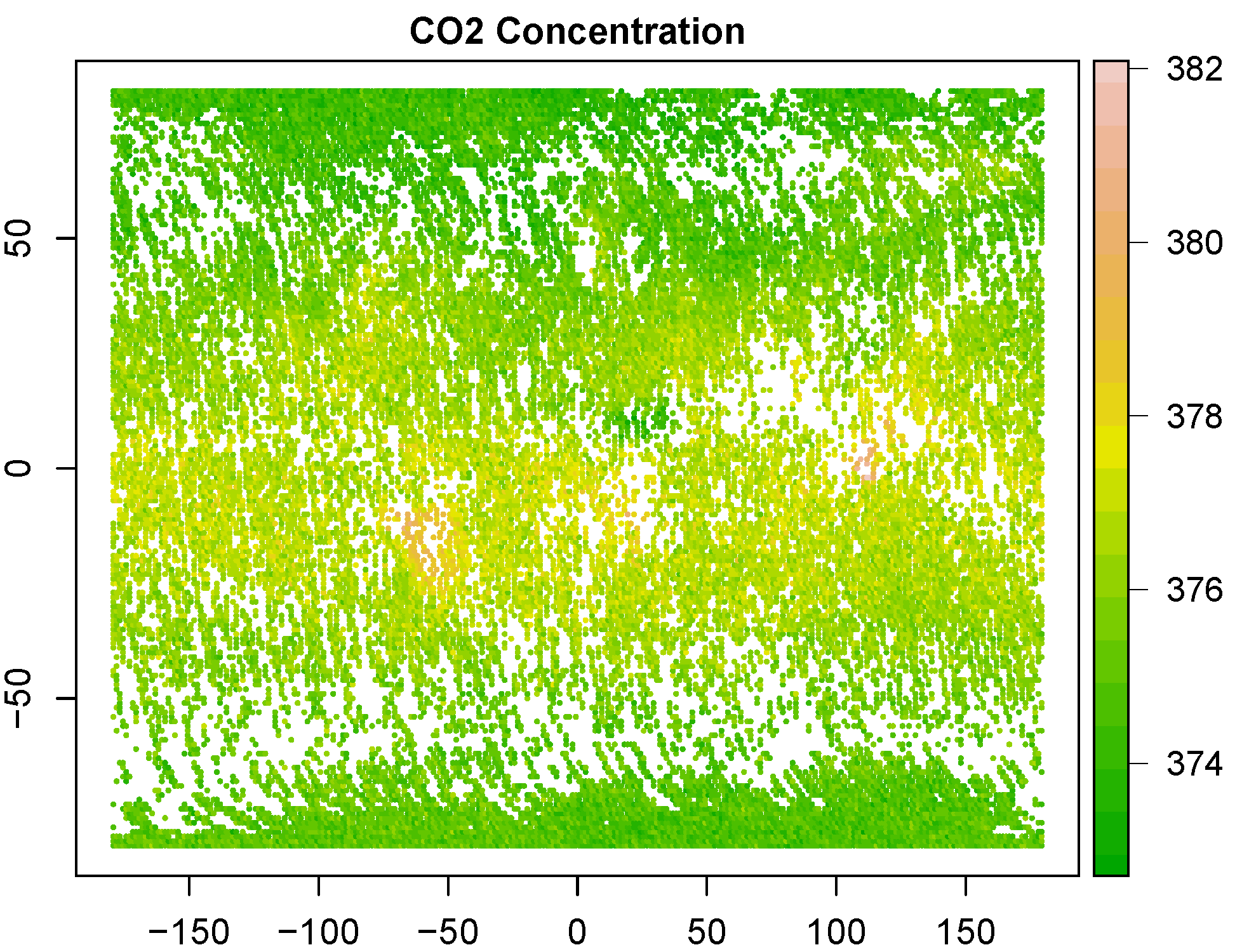}
	\caption{Observations in the datasets.}
	\label{fig:data}
\end{figure}

Figure~\ref{fig:realanalysis} presents the pointwise posterior means of the target functions obtained using each method described in Section~\ref{sec:simulation}, along with the cross-validation (CV) error distribution from a 20-fold CV. The CV errors are computed as $(q^{-1}\sum_{i=1}^q (y_i^\ast -\hat f^\ast(x_i^\ast))^2)^{1/2}$, where $(y_i^\ast,x_i^\ast)_{i=1}^q$ represents the test data in each CV iteration and $\hat f^\ast$ denotes the posterior mean obtained from the corresponding training data. Accordingly, lower CV errors indicate better predictive performance.
An examination of the estimated functions reveals that BBE tends to overlook certain local signals compared to the other methods. In contrast, BART effectively captures these signals, albeit at the expense of smoothness in its estimates. Both the Bayesian P-splines and the proposed method successfully identify local signals while preserving smoothness, with their function estimates appearing comparable. Notably, the proposed method yields the smallest CV errors among all competitors, thereby demonstrating the best predictive performance.

\begin{figure}[t!]
	\begin{subfigure}[b]{\linewidth}
	\centering	
	\includegraphics[width = 1.2in]{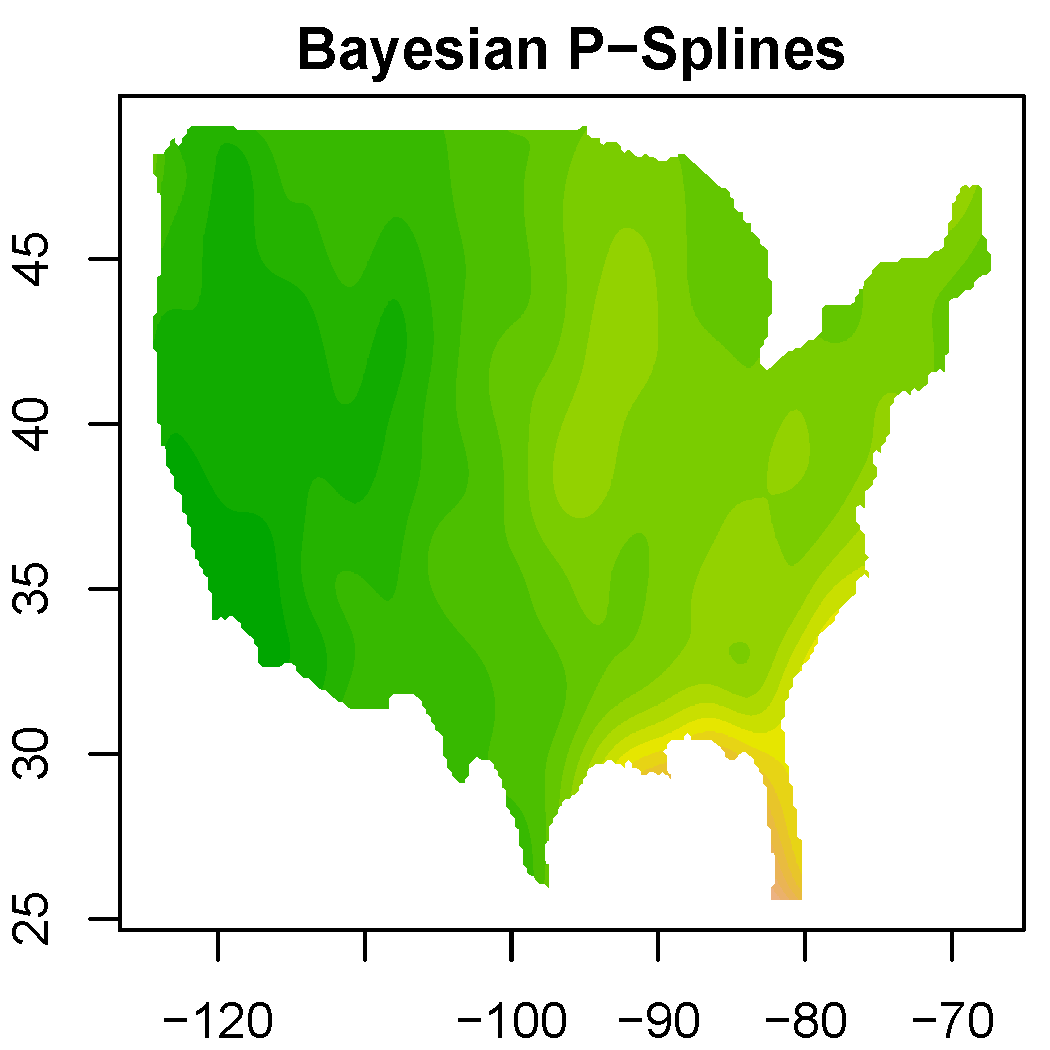}
\includegraphics[width = 1.2in]{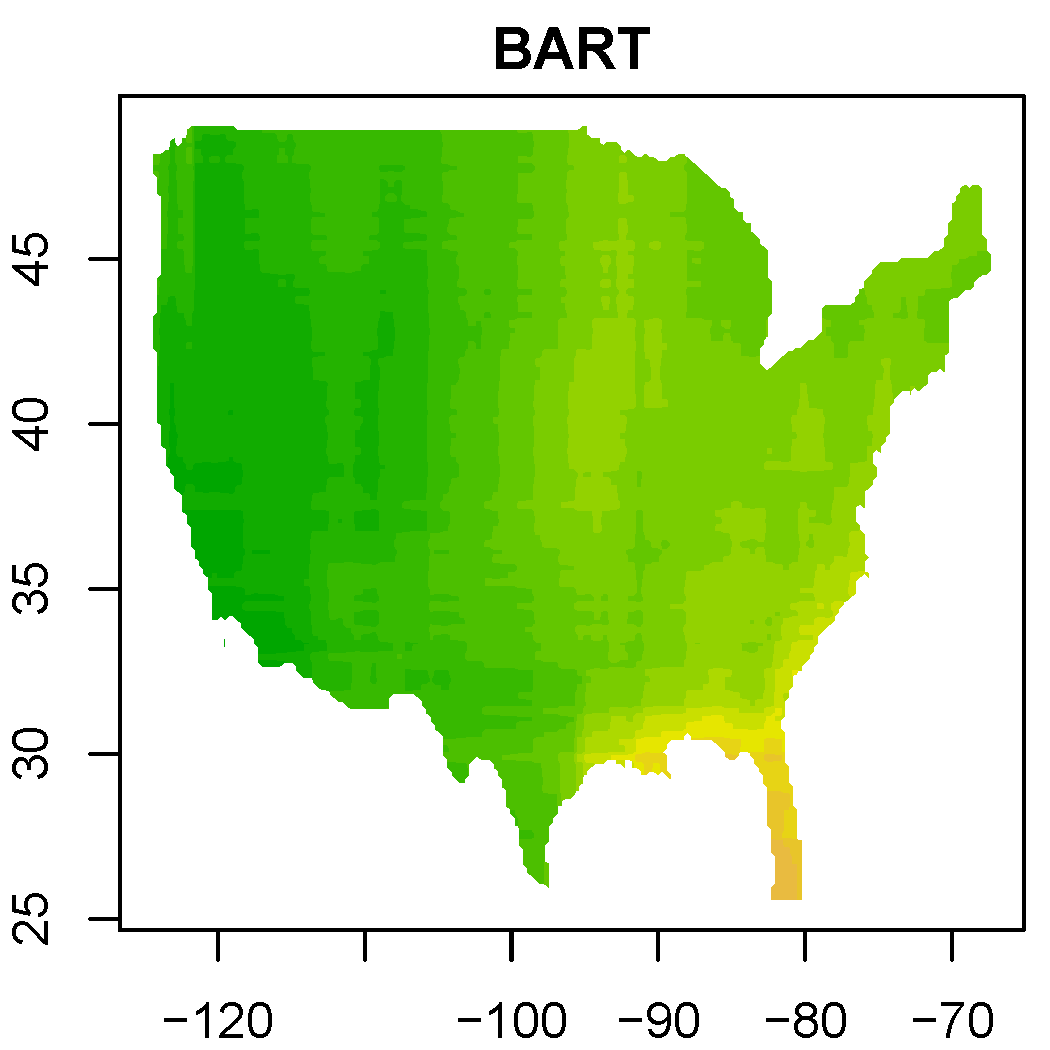}
\includegraphics[width = 1.2in]{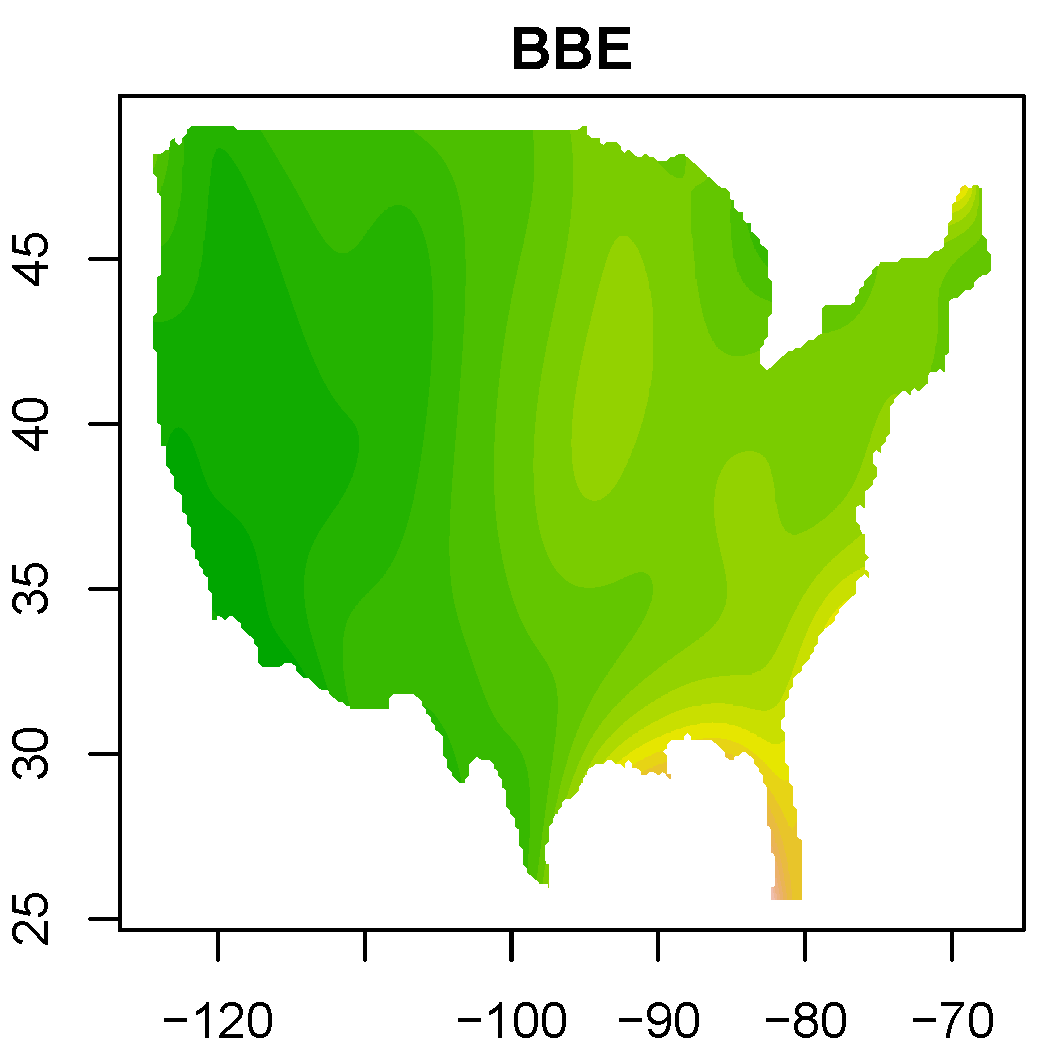}
\includegraphics[width = 1.2in]{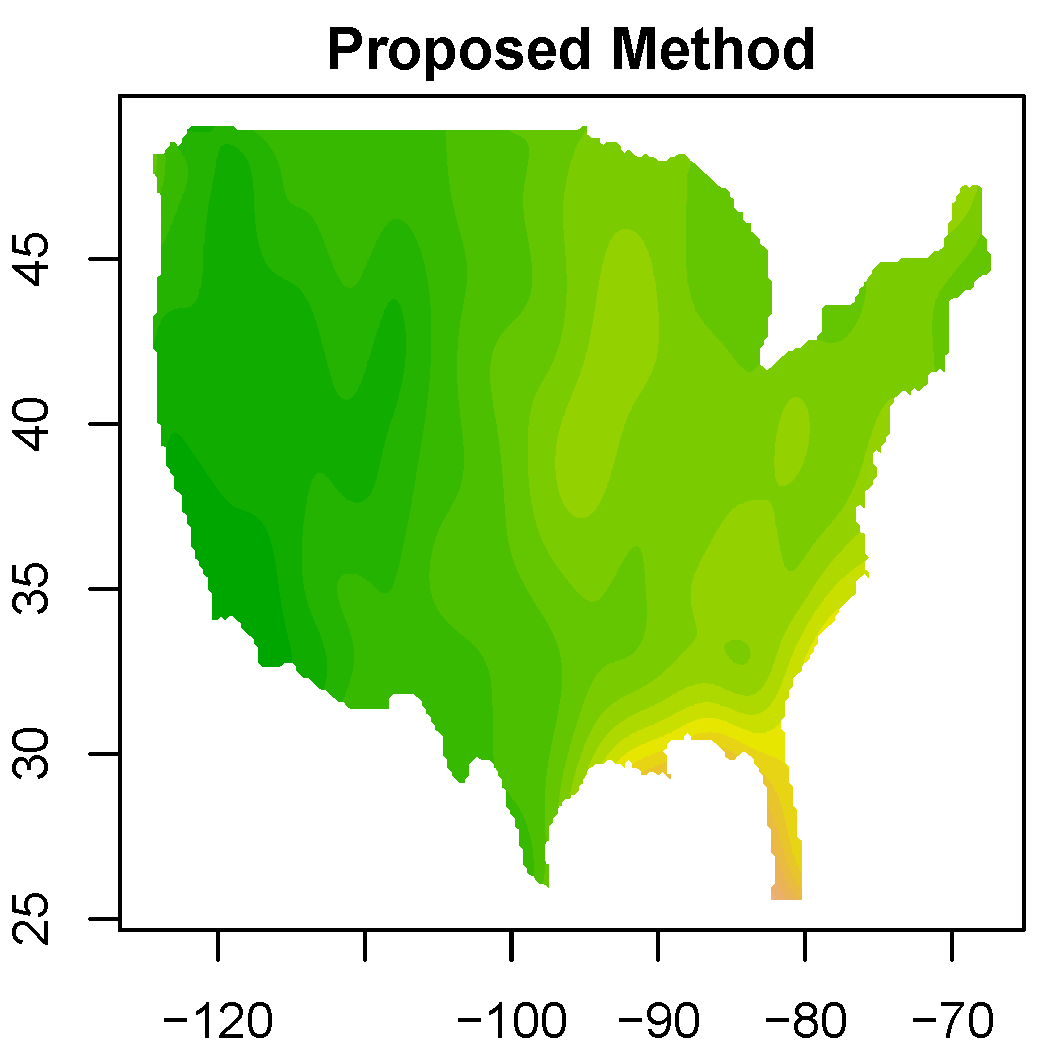}
\includegraphics[width = 1.2in]{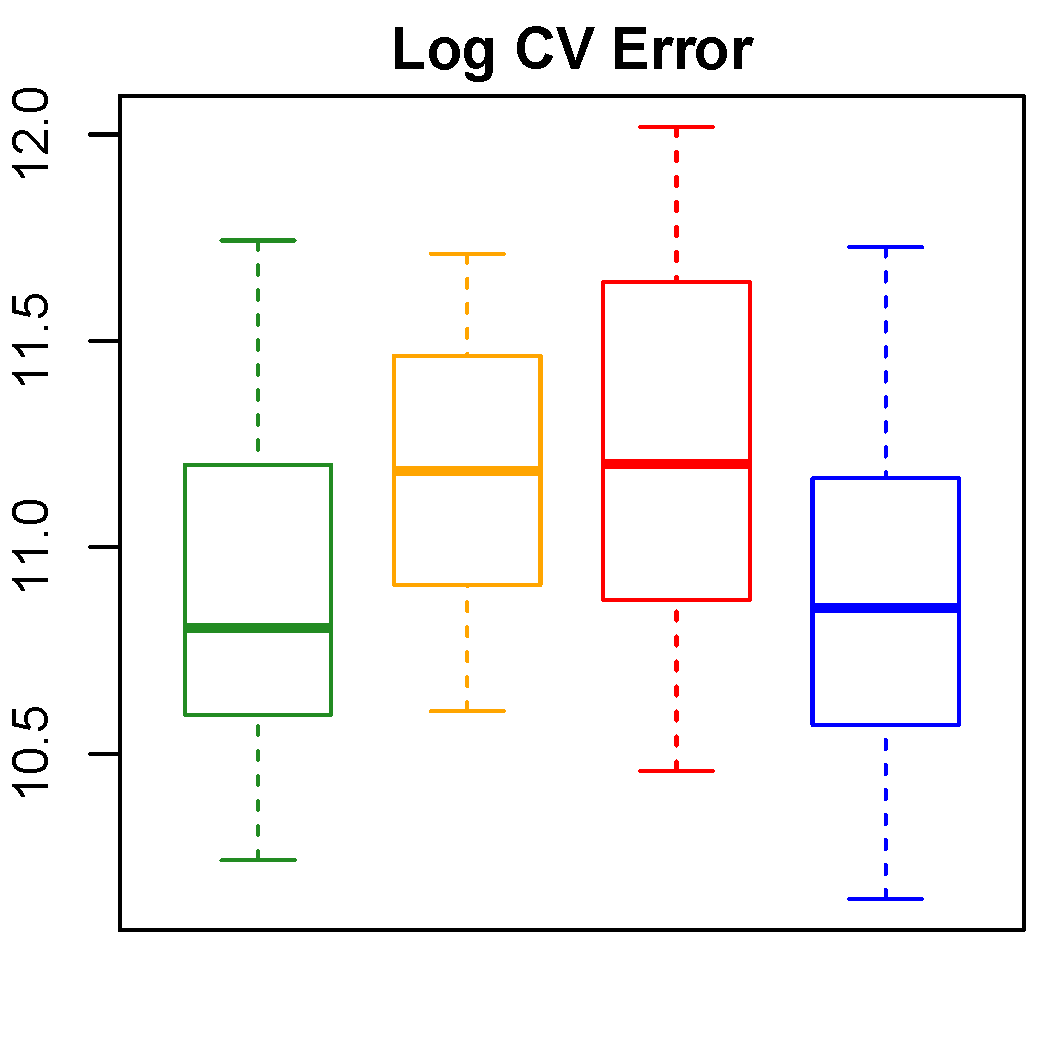}
		\caption{Analysis of Rainfall dataset.}
		\vspace*{0.1in}
	\end{subfigure}
	\begin{subfigure}[b]{\linewidth}
	\centering	
	\includegraphics[width = 1.2in]{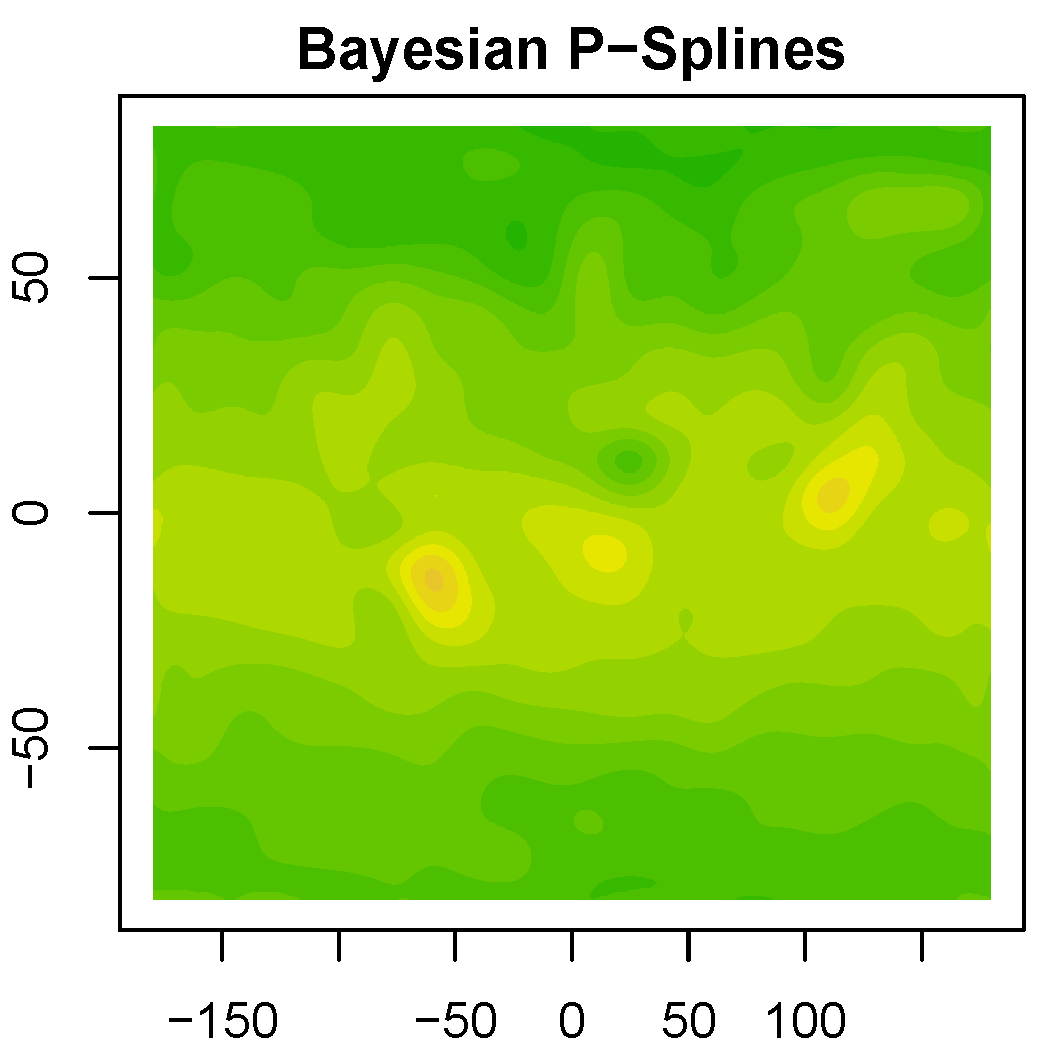}
\includegraphics[width = 1.2in]{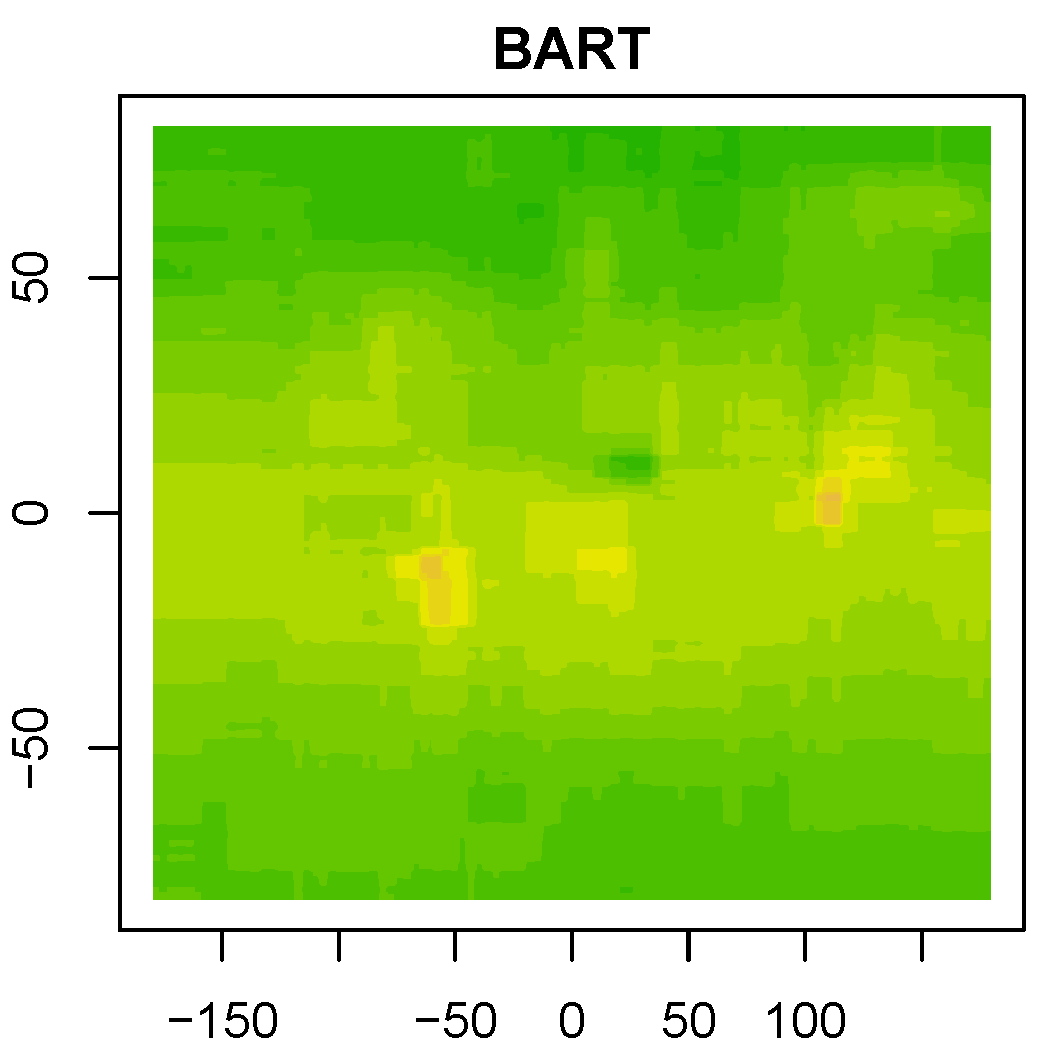}
\includegraphics[width = 1.2in]{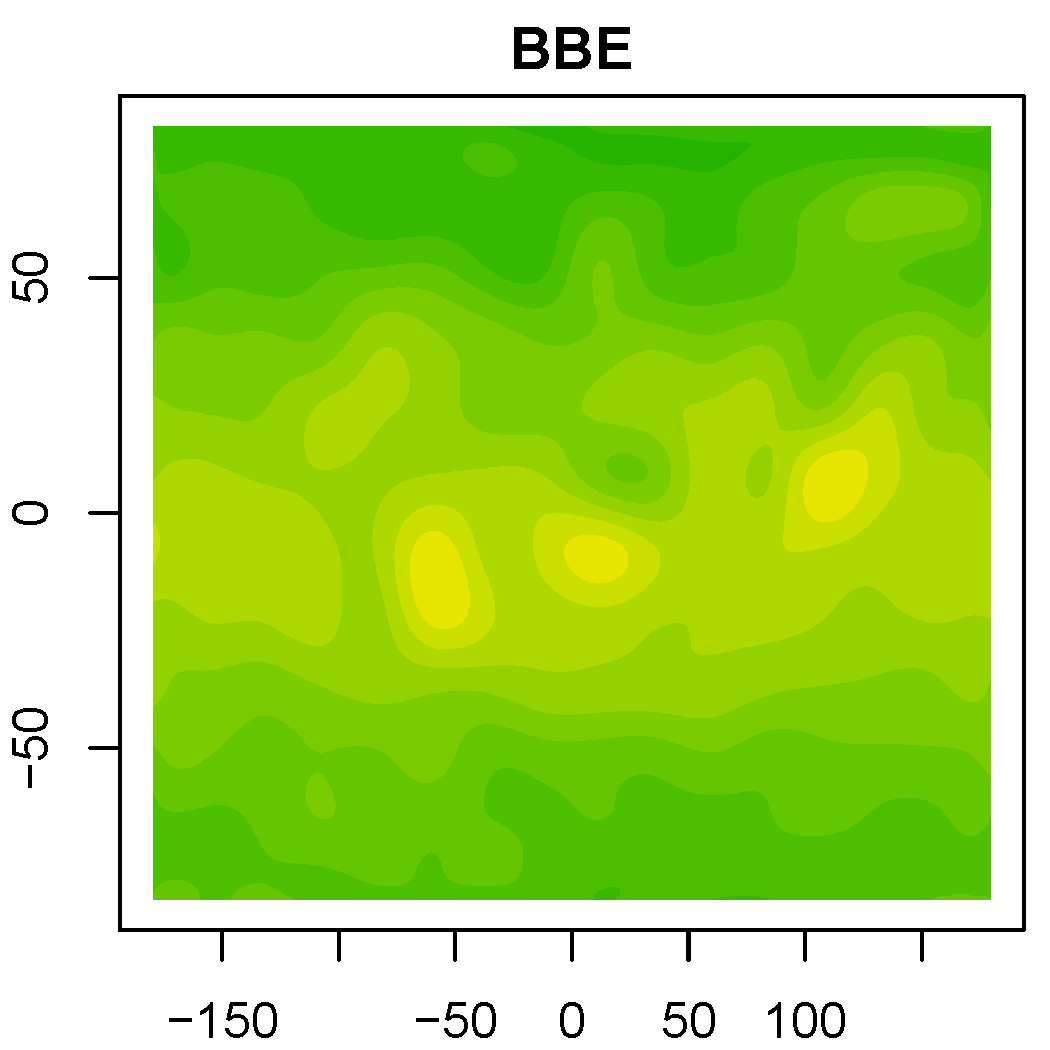}
\includegraphics[width = 1.2in]{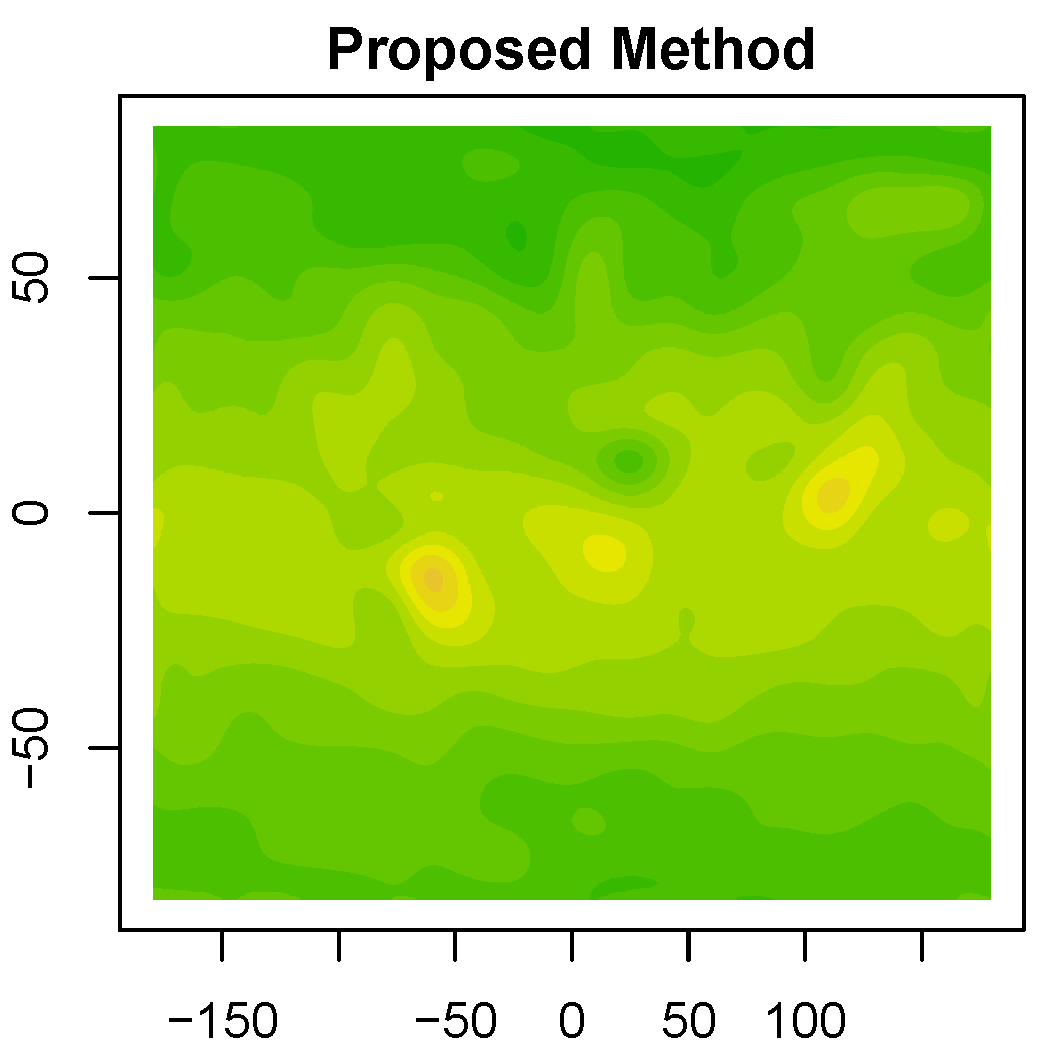}
\includegraphics[width = 1.2in]{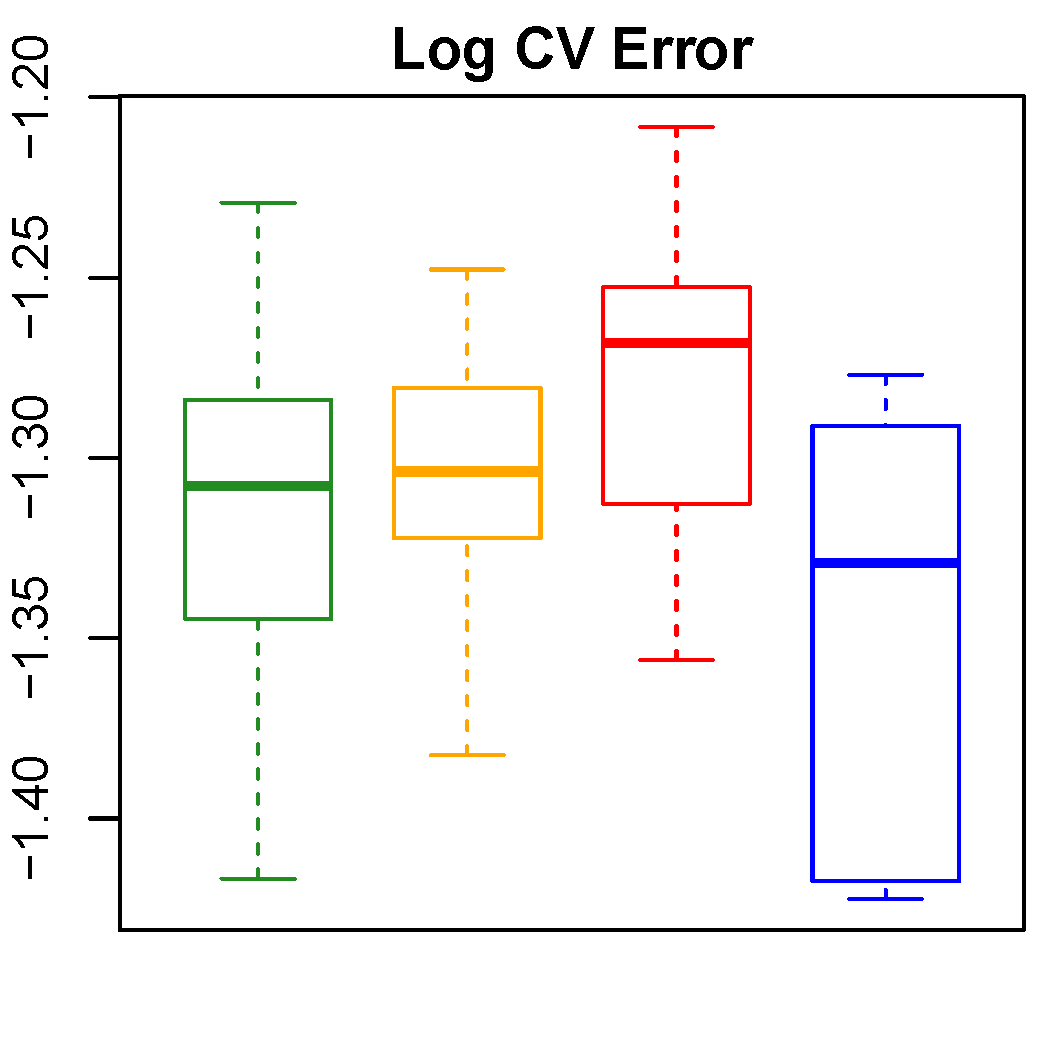}
		\caption{Analysis of CO2 concentration dataset.}
\end{subfigure}
	\caption{Pointwise posterior means of the functions and the logarithm of the 20-fold CV errors for prediction.}
	\label{fig:realanalysis}
\end{figure}

\section{Discussion}
\label{sec:Discussion}

In this paper, we introduce a novel approach to function estimation by employing a penalty-induced prior for basis exploration. While our primary focus is on multivariate nonparametric regression using tensor products of B-splines, the framework readily extends to more structured setups such as additive models and varying coefficient models. As an alternative generalization, one might consider different basis functions in place of B-splines. For instance, natural cubic splines can mitigate the erratic behavior often observed at boundaries with B-spline bases. Our method is inherently flexible and can be easily adapted to incorporate a variety of other basis functions.

\section*{Acknowledgment}
	This research was supported by the National Research Foundation of Korea (NRF) grant funded by the Korean government (MSIT) (2022R1C1C1006735, RS-2023-00217705). 

\appendix
\section{Appendices} \label{appendices}
\addcontentsline{toc}{section}{Appendices}
\renewcommand{\thesubsection}{\Alph{subsection}}

This appendix provides the mathematical proofs for the technical results and details of the MCMC algorithm.
We first establish additional notation.
  For a vector $v=(v_1,\dots,v_m)^T \in \mathbb{R}^{m}$, 
$\lVert v \rVert_2=(\sum_{i=1}^m |v_i|^2)^{1/2}$ denotes its $\ell_2$-norm and $\lVert v \rVert_{\max} = \max_{1\le i\le m} |v_i|$ denotes its maximum norm.
For a function $g:[0,1]^p\rightarrow \mathbb R$, $\lVert g \rVert_\infty = \sup_{u\in[0,1]^p} |f(u)|$ denotes the supremum norm. The maximum eigenvalue of a square matrix is denoted by $\rho_{\max}(\cdot)$. For a positive definite matrix, we often write $\rho_{\max}(\cdot)= \lVert \cdot \rVert_{\text{sp}}$ using the spectral norm, which equals the largest singular value.
For a semi-metric space $(\mathcal E,\delta)$, we denote its $\epsilon$-covering number by $N(\epsilon,\mathcal E, \delta)$.

\subsection{Details of the MCMC Algorithm} \label{appendix_algorithm}

\subsubsection{Grid Sampling for $\tau$}
Grid sampling for $\tau$ is implemented efficiently using \eqref{gridsampling} because matrix operations and eigenvalue computations are required only once to evaluate the posterior density of $\tau$ with many grid points. It remains to verify the expression in \eqref{gridsampling}. Observe that
\begin{align*}
	\log \pi(\tilde\theta_\J \mid \J, \sigma^2, \tau, \lambda) 
	&=
	\frac{1}{2}
	\log \begin{vmatrix} (1-\tau) \tilde P_\J + \tau n^{-1} \tilde B_\J^T \tilde B_\J \end{vmatrix} \\
	&\quad- \frac{1}{2\lambda \sigma^2} \tilde\theta_\J^T \left((1-\tau) \tilde P_\J + \tau n^{-1} \tilde B_\J^T \tilde B_\J\right) \tilde\theta_\J+ C_1,
\end{align*}
where $C_1$ is a constant independent of $\tau$. The determinant term satisfies
\begin{align*}
	\begin{split}
		\log \begin{vmatrix}
			(1-\tau) \tilde P_\J + \tau n^{-1} \tilde B_\J^T \tilde B_\J \end{vmatrix} 
		&= \log \begin{vmatrix}
			I_{J-1} + (\tau n^{-1} \tilde B_\J^T \tilde B_\J)^{-1} 
			(1-\tau) \tilde P_\J
		\end{vmatrix} 
		+
		\log  \begin{vmatrix}
			\tau n^{-1} \tilde B_\J^T \tilde B_\J
		\end{vmatrix}  \\ 
		&= \sum_{k=1}^{J-1} \log \!\left(1 + \frac{n(1-\tau)}{\tau} \rho_k \!\left( (\tilde B_\J^T \tilde B_\J)^{-1} \tilde P_\J \right) \right) + 
		(J-1) \log \tau + C_2,
	\end{split}
\end{align*}
where $C_2$ is a  constant independent of $\tau$. This verifies the expression in \eqref{gridsampling}.

\subsubsection{Slice Sampling for $\lambda$}

The target density is
\begin{align}
\pi(\lambda  \mid  \J, \sigma^2, \theta_0, \tilde\theta_\J, \tau, y) \propto \pi(\lambda) \pi(\tilde\theta_\J  \mid  \J, \sigma^2, \tau, \lambda)
\label{eqn:lambdapost}
\end{align}
The Gaussian density $\pi(\tilde\theta_\J  \mid  \J, \sigma^2, \tau, \lambda)$ is, as a function of $\lambda$, proportional to the inverse gamma density with parameters ${(J-3)}/{2}$ and $\tilde\theta_\J^T[(1-\tau) \tilde P_\J + \tau n^{-1}\tilde B_\J^T \tilde B_\J] \tilde\theta_\J/(2 \sigma^2)$ (recall that $J\ge4$).
For notational simplicity, let $\tilde\pi$ denote the full conditional posterior density of $\lambda$ in \eqref{eqn:lambdapost} and let $g$ denote the inverse gamma density as above.
Observe that
\begin{align*}
	\tilde\pi(\lambda)  &\propto  h(\lambda;c_\lambda) g(\lambda) =  g(\lambda)\int_{0}^{ h(\lambda;c_\lambda)}  d\gamma .
\end{align*}
The joint posterior density of $\lambda$ and the auxiliary variable $\gamma$ is given by
\begin{align*}
	\tilde\pi(\lambda,\gamma) \propto g(\lambda) I(0 < \gamma <  h(\lambda;c_\lambda)).
\end{align*}
Since an exponential density is monotone, the conditional density of each variable is
\begin{align*}
	\tilde\pi(\lambda \mid \gamma) &\propto g(\lambda)  \mathbbm 1\{0 < \lambda < h^{-1}(\gamma;c_\lambda)\}, \\
	\tilde\pi(\gamma \mid \lambda) &\propto  \mathbbm 1\{0 < \gamma <  h(\lambda;c_\lambda)\}.
\end{align*}

\subsection{Proof of the Main Result}
\label{sec:thmproof}


To prove Theorem~\ref{thm:rate}, we adopt the testing-based approach to establish posterior contraction rates \citep{ghosal2000convergence,ghosal2007convergence}. This method requires constructing a test function over a carefully chosen sieve. Because the true variance parameter $\sigma_0$ is unknown, deriving such a test is more complicated than in the standard literature. Here, we present the result from \citet{jeong2025l2norm} without proof, which establishes the existence of a test function for Gaussian models with unknown variance.
Let $\mathbb E_{f,\sigma}$ denote the expectation operator with $f:[0,1]^p\rightarrow \mathbb R$ and $\sigma^2>0$. We also define the semi-metric $d$ as
$d^2((f_1,\sigma_1),(f_2,\sigma_2))=\lVert f_1-f_2 \rVert_n^2 + |\sigma_1-\sigma_2|^2$ for any functions $f_1,f_2$ and $\sigma_1,\sigma_2\in(0,\infty)$.

\begin{lemma}[Theorem~2 of \citet{jeong2025l2norm}]
	\label{lmm:globaltest}
	Consider a positive sequence $\delta_n\rightarrow 0$, and let $\Theta_n$ be a subset of the parameter space for $(f,\sigma)$ satisfying $\log N(\delta_n,\Theta_n,d)\lesssim n\delta_n^2$.
Then, for any sufficiently large constant $M>0$, there exists a test $\varphi_n$ such that
\begin{align*}
	\mathbb E_0 \varphi_n \rightarrow 0,\quad 	\sup_{(f,\sigma)\in\Theta_n:d((f,\sigma),(f_0,\sigma_0))\ge M\delta_n}\mathbb E_{f,\sigma} (1-\varphi_n) \le e^{-Mn\delta_n^2}.
\end{align*}
\end{lemma}


Furthermore, as discussed in Section~\ref{sec:theory}, the marginal $t$-prior for the coefficients exhibits certain challenges. The following lemma restricts our attention to a bounded set of $\sigma^2$ with a suitably increasing order. The proof is provided in Appendix~\ref{appendix_proof}.

\begin{lemma}
\label{lmm:sigma}
Under the conditions for Theorem~\ref{thm:rate}, the marginal posterior of $\sigma^2$ satisfies $\mathbb E_0\Pi\{\sigma>D_n\mid y\}\rightarrow 0$ for any $D_n\rightarrow \infty$.
\end{lemma}

We further establish eigenvalue bounds for the prior covariance matrix corresponding to the coefficients of the original B-spline basis. Although a prior distribution is assigned to the transformed B-spline basis, our theoretical analysis is based on the original construction for convenience. This approach ensures that the prior allocates sufficient probability mass around the true values while maintaining exponentially small tails. The proof is provided in Appendix~\ref{appendix_proof}.

\begin{lemma}
	 \label{lmm:eigen}
	 Let $\Sigma_{\J,\sigma,\tau,\lambda}$ be the prior variance of the B-spline coefficients $\theta_\J$ induced from \eqref{eqn:Covariance_proposedmethod}. We obtain 
	\begin{align*}
	\rho_{\min}(\Sigma_{\J,\sigma,\tau,\lambda}^{-1})&\ge \sigma^{-2}\min(\kappa^{-2}, \tau/\lambda) \rho_{\min}(n^{-1} B_\J^T B_\J), \\
	\rho_{\max}(\Sigma_{\J,\sigma,\tau,\lambda}^{-1})&\le 16p(\lambda \sigma^2)^{-1}(1-\tau)+ \sigma^{-2}\max(\kappa^{-2}, \tau/\lambda) .
\end{align*}
\end{lemma}

Now, we present the proof of Theorem~\ref{thm:rate}. 

\begin{proof}[Proof of Theorem~\ref{thm:rate}.]
Observe that for any $M>0$,
\begin{align}
	\begin{split}
	&\Pi\{\lVert f-f_0 \rVert_n+|\sigma-\sigma_0|>M\epsilon_n\mid y\}\\
&\quad\le \Pi\{\sigma>n\mid y\}+\Pi\{\lVert f-f_0 \rVert_n+|\sigma-\sigma_0|>M\epsilon_n\mid \sigma< n,  y\}.
	\end{split}
	\label{eqn:post1}
\end{align}
The expected value of the first term on the right-hand side goes to zero by Lemma~\ref{lmm:sigma}. Let $\widetilde\Pi(\cdot) = \Pi(\cdot \mid \sigma<n)$ be the prior that is renormalized and restricted to $\{\sigma^2:\sigma>n\}$. Denoting the induced posterior by $\widetilde \Pi(\cdot \mid y)$, the second term in \eqref{eqn:post1} can be expressed as
$$
\widetilde\Pi\{\lVert f-f_0 \rVert_n+|\sigma-\sigma_0|>M\epsilon_n\mid y\}.
$$
We will verify that the expected value of this expression goes to zero. 
	For the Kullback-Leibler divergence $K(p_1,p_2)=\int \log(p_1/p_2)p_1$ and its second order variation $V(p_1,p_2)=\int |\log(p_1/p_2)-K(p_1,p_2)|^2 p_1$, define 
\begin{align*}
	\mathcal A_n=\left\{ (f,\sigma): \sum_{i=1}^n K(p_{0,i},p_{f,\sigma, i})\le n\epsilon_n^2 , \, \sum_{i=1}^n V(p_{0,i},p_{f,\sigma,i})\le n\epsilon_n^2\right\},
\end{align*}
where $p_{0,i}$ and $p_{f,\sigma, i}$ denote the Gaussian densities of the individual $y_i$ using $(f_0,\sigma_0)$ and $(f,\sigma)$, respectively.
We define $\mathcal F_\J= \left\{f = \theta_\J^T\psi_\J:  \theta_\J \in\mathbb R^J\right\}$ and $\mathcal F=\bigcup_{\J:J_m\ge l_m,m=1,\dots,p} \mathcal F_\J$, where $\psi_\J:[0,1]^p\rightarrow \mathbb R^J$ denotes the tensor product basis functions.
In other words, $\mathcal F_\J$ is the $p$-dimensional spline space with $\J$ and $\mathcal F$ is their collection for all $\J$.
Following the testing-based approach \citep{ghosal2000convergence,ghosal2007convergence}, we only need to show that there exists $\tilde{\mathcal F}_n\subset \mathcal F$ such that for some constant $c>0$,
\begin{align}
	\widetilde \Pi(\mathcal A_n)&\ge e^{-cn\epsilon_n^2},\label{eqn:cond1}\\
	\mathbb E_0 \varphi_n \rightarrow 0,\quad 	\sup_{(f,\sigma)\in\tilde{\mathcal F}_n\times (0,n): d((f,\sigma),(f_0,\sigma_0))\ge M\epsilon_n}\mathbb E_{f,\sigma} (1-\varphi_n) &\le e^{-(c+4)n\epsilon_n^2}. \label{eqn:cond2}\\
	\widetilde\Pi((f,\sigma)\notin\tilde{\mathcal F}_n\times (0,n))&\le e^{-(c+4)n\epsilon_n^2}.
	\label{eqn:cond3}
\end{align}
 (This will verify that the posterior contraction for $(f,\sigma)$ is $\epsilon_n$ with respect to $d$, but $\sqrt{\lVert\cdot\rVert_n^2+|\cdot|^2}$ and $\lVert\cdot\rVert_n+|\cdot|$ have the same order.)
Below we will show that the conditions in \eqref{eqn:cond1}--\eqref{eqn:cond3} are satisfied with $\epsilon_n \asymp ((\log n)/n)^{\bar\alpha/(2\bar\alpha+p)}$.

	We first verify \eqref{eqn:cond1}.
By direct calculations, 
\begin{align*}
	\frac{1}{n}\sum_{i=1}^n K(p_{0,i},p_{f,\sigma,i})&=\frac{1}{2} \log \left( \frac{\sigma^2}{\sigma_0^2} \right)-\frac{1}{2}\left(1-\frac{\sigma_0^2}{\sigma^2}\right)+\frac{\lVert f-f_0 \rVert_n^2}{2\sigma^2},\\
	\frac{1}{n}\sum_{i=1}^n V(p_{0,i},p_{f,\sigma,i}) &=\frac{1}{2}\left(1-\frac{\sigma_0^2}{\sigma^2}\right)^2+\frac{\sigma_0^2\lVert f-f_0 \rVert_n^2}{\sigma^4}.
\end{align*}
Using the Taylor expansion, it is easy to see that, for any $\epsilon_n\rightarrow 0$,
there exists a constant $C_1>0$ such that
\begin{align*}
	\mathcal A_n&\supset \{(f,\sigma): \lVert f-f_0\rVert_n\le C_1 \epsilon_n , | \sigma^2-\sigma_0^2|\le C_1 \epsilon_n \}.
\end{align*}
Therefore,
\begin{align*}
	\widetilde \Pi(\mathcal A_n)&\ge\widetilde \Pi\{(f,\sigma): \lVert f-f_0\rVert_n\le C_1\epsilon_n , |\sigma^2-\sigma_0^2|\le C_1\epsilon_n\}\\
	&= \int_{\{\sigma^2:|\sigma^2-\sigma_0^2|\le C_1\epsilon_n\}}\Pi\{f\in\mathcal F:\lVert f-f_0\rVert_n\le C_1\epsilon_n\mid \sigma\}d\widetilde\Pi(\sigma^2)\\
	&\ge \widetilde\Pi\{|\sigma^2-\sigma_0^2|\le C_1\epsilon_n\}\inf_{\sigma^2:|\sigma^2-\sigma_0^2|\le C_1\epsilon_n} \Pi\{f\in\mathcal F:\lVert f-f_0\rVert_n\le C_1\epsilon_n\mid \sigma^2\}.
\end{align*}
Since the density of a proper inverse gamma distribution is bounded away from zero on a compact subset of $(0,\infty)$, we obtain that
\begin{align*}
\log  \widetilde\Pi\{|\sigma^2-\sigma_0^2|\le C_1\epsilon_n\}\ge  \log \Pi\{|\sigma^2-\sigma_0^2|\le C_1\epsilon_n\}\gtrsim\log \epsilon_n\gtrsim -\log n.
\end{align*}
By restricting $J_m$ to $\hat J_m\asymp (n/\log n)^{\bar\alpha/(\alpha_m(2\bar\alpha+p))}$, we obtain
\begin{align*}
   &\inf_{\sigma^2:|\sigma^2-\sigma_0^2|\le C_1\epsilon_n} \Pi\{f\in\mathcal F:\lVert f-f_0\rVert_n\le C_1\epsilon_n\mid \sigma^2\}\\
    &\quad \ge \Pi\{\J=\hat \J\}\inf_{\sigma^2:|\sigma^2-\sigma_0^2|\le C_1\epsilon_n} \Pi\{f\in\mathcal F_{\J}:\lVert f-f_0\rVert_\infty\le C_1\epsilon_n\mid \J=\hat \J, \sigma^2\},
\end{align*}
where $\hat\J = \{\hat J_m\}_{m=1}^p$. Using Assumption~\ref{asm:prior}, we obtain that $\log\Pi\{\J=\hat \J\}\gtrsim - \hat J\log \hat J \gtrsim -n\epsilon_n^2$ since $\hat J\log n\asymp n\epsilon_n^2$, where $\hat J = \prod_{m=1}^p \hat J_m \asymp (n/\log n)^{p/(2\bar\alpha+p)}$. Moreover, the classical B-spline theory \citep{schumaker2007spline} shows that for any $f_0\in\mathcal H^\alpha([0,1])$, there exists $\hat f = \hat\theta_\J^T \psi_\J\in\mathcal F_\J$ with $\lVert\hat\theta_\J\rVert_{\max}\lesssim 1$ such that
$$
 \lVert \hat f-f_0\rVert_\infty \lesssim \sum_{m=1}^p J_m^{-\alpha_m}.
$$
The right-hand side is order of $\epsilon_n$ if $\J=\hat\J$. 
Since the inequality $\lVert \theta_\J^T \psi_\J\rVert_\infty\le \lVert \theta_\J\rVert_{\max}$ holds by the sum-to-unity property of B-splines, there exists $C_2>0$ and $C_3>0$ such that
\begin{align*}
	&\Pi\{f\in\mathcal F_{\J}:\lVert f-f_0\rVert_\infty\le C_1\epsilon_n\mid \J=\hat \J,\sigma^2\} \\
	&\quad\ge \Pi\{\theta_\J\in\mathbb R^{J}: \lVert \theta_\J - \hat \theta_\J  \rVert_{\max}\le C_2\epsilon_n \mid \J=\hat \J,\sigma^2 \}\\
	&\quad\ge \Pi\!\left\{\theta_\J\in\mathbb R^{J}: \lVert \Sigma_{\J,\sigma,\tau,\lambda}^{-1/2}( \theta_\J - \hat {\theta}_\J ) \rVert_2\le \frac{C_2\epsilon_n}{\lVert \Sigma_{\J,\sigma,\tau,\lambda}\rVert_{\text{sp}}^{1/2}}\mid  \J=\hat \J,\sigma^2\right \}\\
	&\quad\gtrsim\frac{1}{n^{C_3}}\inf_{\lambda\in(\lambda_0,n^k)} \inf_{\tau\in(n^{-k},1)}\Pi\!\left\{\theta_\J\in\mathbb R^{J}: \lVert \Sigma_{\J,\sigma,\tau,\lambda}^{-1/2}(\theta_\J - \hat {\theta}_\J ) \rVert_2\le \frac{C_2\epsilon_n}{\lVert \Sigma_{\J,\sigma,\tau,\lambda}\rVert_{\text{sp}}^{1/2}}\mid \J=\hat \J, \sigma^2, \tau,\lambda \right\},
\end{align*}
where the last inequality holds because $\log\Pi\{\lambda_0 < \lambda < n^{k}\} \gtrsim -\log n$ and 
$\Pi\{\tau > n^{-k}\}\gtrsim 1$ follow from Assumption~\ref{asm:prior}.
Since $ \Sigma_{\J,\sigma,\tau,\lambda}^{-1/2}\theta_\J $ has the Gaussian prior with zero mean and identity covariance matrix, using the change of variables, we obtain that for any $v>0$,
\begin{align*}
	&\Pi\!\left\{\theta_\J\in\mathbb R^{J}: \lVert \Sigma_{\J,\sigma,\tau,\lambda}^{-1/2}( \theta_\J - \hat \theta_\J ) \rVert_2\le v\mid \J=\hat \J,\sigma^2, \tau,\lambda\right \} \\
	&\quad\ge
	2^{-\hat J/2}\exp\!\left(-\big\lVert \Sigma_{\hat \J,\sigma,\tau,\lambda}^{-1/2}\hat \theta_{\hat \J} \big\rVert_2^2\right)\Pi\!\left\{\theta_\J\in\mathbb R^{J}: \lVert \Sigma_{\J,\sigma,\tau,\lambda}^{-1/2}\theta_\J \rVert_2\le v/\sqrt{2}\mid \J=\hat \J,\sigma^2, \tau,\lambda \right\}.
\end{align*}
(For example, see the calculation on page 216 of \citet{ghosal2007convergence}.) By Lemma~\ref{lmm:eigen}, observe that
$$
\big\lVert \Sigma_{\hat \J,\sigma,\tau,\lambda}^{-1/2}\hat \theta_{\hat \J} \big\rVert_2^2 \le \big\lVert \Sigma_{\hat \J,\sigma,\tau,\lambda}^{-1}\big\rVert_{\text{sp}}\big\lVert\hat \theta_{\hat \J} \big\rVert_2^2 \le \hat J\big\lVert \Sigma_{\hat \J,\sigma,\tau,\lambda}^{-1}\big\rVert_{\text{sp}}\big\lVert\hat \theta_{\hat \J} \big\rVert_{\max}^2 \lesssim \hat J/\sigma^2,
$$
uniformly over $\lambda\in(\lambda_0,n^k)$ and $\tau\in(n^{-k},1)$.
Using the fact that $\lVert \Sigma_{\J,\sigma,\tau,\lambda}^{-1/2}\theta_\J \rVert_2^2$ has a $\chi^2$-distribution with $J$ degrees of freedom, we obtain
\begin{align}
	\begin{split}
	\Pi\!\left\{\theta_\J\in\mathbb R^{J}: \lVert \Sigma_{\J,\sigma,\tau,\lambda}^{-1/2} \theta_\J \rVert_2\le v/\sqrt{2}\mid \J=\hat \J,\sigma^2,\tau,\lambda \right \}&\ge \frac{2/\hat J }{2^{\hat J }\Gamma(\hat J /2)}v^{\hat J }e^{-v^2/4}\\
&\ge e^{-C_4(\hat J \log \hat J +\hat J \log (1/v)+v^2)},
	\end{split}
	\label{eqn:lowbo1}
\end{align}
for some $C_4>0$. Observe that by Lemma~\ref{lmm:eigen},
\begin{align}
	\begin{split}
	\lVert\Sigma_{\hat{\J},\sigma,\tau,\lambda}\rVert_{\text{sp}}&\ge 1/\rho_{\max}\big(\Sigma_{\hat{\J},\sigma,\tau,\lambda}^{-1}\big)	\gtrsim \sigma^2/(1+1/\lambda)\\
\lVert\Sigma_{\hat{\J},\sigma,\tau,\lambda}\rVert_{\text{sp}}  &= 1/	\rho_{\min}\big(\Sigma_{\hat{\J},\sigma,\tau,\lambda}^{-1}\big)\lesssim  \sigma^2 (1+\lambda/\tau)/\rho_{\min}(n^{-1} B_{\hat{\J}}^T B_{\hat{\J}}).
	\end{split}
	\label{eqn:eigenbo}
\end{align}
Assumption~\ref{asm:eigen} further implies that $\log \rho_{\min}(n^{-1} B_{\hat{\J}}^T B_{\hat{\J}})\gtrsim -\log n$.
Therefore, by plugging in $v=C_2\epsilon_n/\lVert \Sigma_{\hat\J,\sigma,\tau,\lambda}\rVert_{\text{sp}}^{1/2}$ and applying the lower and upper bounds in \eqref{eqn:eigenbo},
one deduces that the logarithm of \eqref{eqn:lowbo1} is bounded below by a constant multiple of $ -\hat J\log \hat J-\hat J\log n-\hat J\log\sigma-\sigma^{-2}$, uniformly over $\lambda\in(\lambda_0,n^k)$ and $\tau\in(n^{-k},1)$.
Combining the bounds, it follows that
\begin{align*}
	&\log \inf_{\sigma^2:|\sigma^2-\sigma_0^2|\le C_1\epsilon_n} \Pi\{f\in\mathcal F_\J:\lVert f-f_0\rVert_\infty\le C_1\epsilon_n\mid \J=\hat \J,\sigma^2\}
	\gtrsim -\hat J \log n.
\end{align*}
Putting everything together, we obtain that $\widetilde \Pi(\mathcal A_n)\ge e^{- cn\epsilon_n^2}$ for some $c>0$.

 Next, to verify \eqref{eqn:cond2}, Lemma~\ref{lmm:globaltest} implies that it suffices to show that $\log N(\epsilon_n,\tilde{\mathcal F}_n\times(0,n), d)\lesssim n\epsilon_n^2$.
 Let 
 $$
\tilde{\mathcal F}_n=\bigcup_{\J: J\le \bar J}\{f\in\mathcal F_\J: \lVert \theta_\J\rVert_{\max}\le n^{D_1}\},
 $$
 for a constant $D_1>0$, where $\bar J = \lfloor D_2 \hat J \rfloor$ for a constant $D_2>0$. 
Since $\lVert \theta_\J^T \psi_\J \rVert_n\le \lVert\theta_\J\rVert_{\max}$ by the sum-to-unity property, we obtain that
 \begin{align*}
\log N(\epsilon_n,\tilde{\mathcal F}_n\times (0,n), d) &\le \log N\!\left(\frac{\epsilon_n}{2},\tilde{\mathcal  F}_n, \lVert\cdot\rVert_n\right) +\log N\!\left(\frac{\epsilon_n}{2},(0,n), |\cdot|\right) \\
&\le \log \!\left[ \sum_{\J: J\le \bar J} N\!\left( \frac{\epsilon_n}{2}, \{\theta_\J\in\mathbb R^J:\lVert \theta_\J\rVert_{\max}\le n^{D_1}\},\lVert\cdot \rVert_{\max}
\right)\right] + \log (2n/\epsilon_n)\\
&\le \log\#\{\J: J\le \bar J\}+\bar J\log( 4 n^{D_1} /\epsilon_n)+\log (2n/\epsilon_n).
 \end{align*}
Since $\#\{\J: J\le \bar J\}\le \#\{\J: J_m\le \bar J,m=1,\dots,p\} \le\bar J^{p}$, the rightmost side is bounded by a multiple of $\bar J \log n\lesssim n\epsilon_n^2$.

Lastly, to verify \eqref{eqn:cond3}, observe that
\begin{align*}
\widetilde\Pi\{(f,\sigma)\notin\tilde{\mathcal  F}_n\times (0,n)\}&\le \Pi\{\J: J> \bar J\}+\Pi\{\tau < n^{-k}\}+\Pi\{\lambda>n^k\}\\
&\quad+\widetilde\Pi\{\lVert\theta_\J \rVert_{\max} > n^{D_1},  J\le \bar J,\tau \ge n^{-k}, \lambda\le n^k\}.
\end{align*}
Assumption~\ref{asm:prior} implies that $\log\Pi\{\tau < n^{-k}\}\lesssim -n$ and $\log\Pi\{\lambda > n^k\}\lesssim -n$. 
Furthermore, for any $j\gtrsim (n/\log n)^{p/(2\bar\alpha+p)} $, Assumption~\ref{asm:prior} also shows that 
$$
\log\Pi\{\J:J = j\} =\log\sum_{\J:J = j}\pi(\J)\le \log\#\{\J:J = j\}- C_5 j\log j \le p\log j  - C_5 j\log j\lesssim  -j\log j,
$$
for some constant $C_5>0$.
By applying the sum of a geometric series, we obtain
$\log\Pi\{\J:J > j\}\lesssim -j\log j$ for such $j$. Therefore, if $D_2$ is chosen to be sufficiently large, it follows that $\Pi\{\J: J> \bar J\}+\Pi\{\tau < n^{-k}\}+\Pi\{\lambda>n^k\}=o(e^{-(c+2)n\epsilon_n^2})$.
Moreover,
\begin{align}
	\begin{split}
	&\widetilde\Pi\{\lVert\theta_\J \rVert_{\max} > n^{D_1},J\le \bar J,\tau\ge n^{-k},\lambda\le n^k\}\\
&\quad=\sum_{\J: J\le \bar J}\pi(\J)\int_0^{n^k} \int_{n^{-k}}^1\int_0^{n}\Pi\{\lVert \theta_\J\rVert_{\max} > n^{D_1} \mid \J,\sigma^2,\tau,\lambda\}d\widetilde\Pi(\sigma^2)d\Pi(\tau)d\Pi(\lambda) \\
&\quad\le \max_{\J: J\le \bar J}\sup_{\sigma\in(0, n)}\sup_{\tau\in(n^{-k},1)}\sup_{\lambda\in(0, n^k)}  \Pi\{\lVert \theta_\J\rVert_{\max} > n^{D_1}  \mid \J,\sigma^2,\tau,\lambda\}\\
&\quad\le \max_{\J: J\le \bar J}\sup_{\sigma\in(0, n)}\sup_{\tau\in(n^{-k},1)}\sup_{\lambda\in(0, n^k)}   \Pi\!\left\{\lVert \Sigma_{\J,\sigma,\tau,\lambda}^{-1/2}\theta_\J\rVert_{\max} > n^{D_1} / \sqrt{J\lVert \Sigma_{\J,\sigma,\tau,\lambda}\rVert_{\text{sp}}} \mid \J,\sigma^2,\tau,\lambda\right\}\\
&\quad\le \sup_{\sigma\in(0, n)}\sup_{\tau\in(n^{-k},1)}\sup_{\lambda\in(0, n^k)}  2\bar 
J\exp\!\left(-\frac{n^{2D_1} }{ 2\bar J\max_{\J: J\le \bar J}\lVert\Sigma_{ \J,\sigma,\tau,\lambda}\rVert_{\text{sp}}}\right).
	\end{split}
\label{eqn:sievepr}
\end{align}
Using Lemma~\ref{lmm:eigen} and Assumption~\ref{asm:eigen}, we obtain
$$
\max_{\J: J\le \bar J}\log \lVert \Sigma_{\J,\sigma,\tau,\lambda}\rVert_{\text{sp}} = 1/\min_{\J: J\le \bar J}\log\rho_{\min}(\Sigma_{\J,\sigma,\tau,\lambda}^{-1})\lesssim \log \sigma + \log \max\{\kappa^2,\lambda/\tau\} +\log n.
$$
Therefore, with a suitably chosen $D_1$, the rightmost side of \eqref{eqn:sievepr} is bounded by $e^{-n}$, which verifies \eqref{eqn:cond3}.
\end{proof}

\subsection{Proofs of the Remaining Technical Results} \label{appendix_proof}

\begin{proof}[Proof of Lemma~\ref{lmm:pspline}]
	We first find a nonsingular matrix $A_\J\in\mathbb R^{J\times J}$ such that $A_\J\theta_\J = (\theta_0,\tilde\theta_\J^T)^T$ with $B_\J\theta_\J = \theta_0 1_n+\tilde B_\J\tilde\theta_\J$. Let $B_\J$ be partitioned as $B_\J=[B_{\J,1}, B_{\J,-1}]$ with $B_{\J,1}\in\mathbb R^n$ and $B_{\J,-1}\in\mathbb R^{n\times (J-1)}$, and let $\theta_\J$ be partitioned as $\theta_\J=(\theta_{\J,1},\theta_{\J,-1}^T)^T$ with $\theta_{\J,1}\in\mathbb R$ and $\theta_{\J,-1}\in\mathbb R^{J-1}$.
	Using the projection matrix $H_1=n^{-1}1_n1_n^T$ and the sum-to-unity property $B_\J1_J = 1_n$ of B-splines, we obtain
	\begin{align*} 
	B_\J\theta_\J &= H_1 B_\J \theta_\J + (I_n-H_1) (\theta_{\J,1} B_{\J,1}+ B_{\J,-1} \theta_{\J,-1} ) \\
	&= H_1 B_\J \theta_\J + (I_n-H_1) (\theta_{\J,1} (1_n-B_{\J,-1}1_{J-1})+ B_{\J,-1} \theta_{\J,-1} )\\
	&= H_1 B_\J \theta_\J + \tilde B_\J(\theta_{\J,-1}-\theta_{\J,1}1_{J-1}),
	\end{align*}
where $\tilde B_\J=(I_n-H_1) B_{\J,-1} $. This gives the linear map
\begin{align}  
	A_\J = 
	\begin{pmatrix}
n^{-1} 1_n^T B_{\J,1} & n^{-1} 1_n^T B_{\J,-1}\\
-1_{J-1} & I_{J-1}
	\end{pmatrix} \in \mathbb{R}^{J \times J}.
	\label{eqn:linearmap}
\end{align} 
Define $S_{J_m}=I_{J_1}\otimes\dots\otimes I_{J_{m-1}}\otimes D_{J_m} \otimes I_{J_{m+1}}\otimes\dots\otimes I_{J_p}\in\mathbb R^{(J-2)\times J}$ so that $P_{J_m} = S_{J_m}^T S_{J_m}$. Observe that for any matrices $G_1\in\mathbb R^{p_1\times q_1}$ and $G_2\in\mathbb R^{p_2\times q_2}$, we have $(G_1\otimes G_2)1_{q_1 q_2} = \text{vec}(G_2 1_{q_2} 1_{q_1}^T G_1^T)$. This implies that $S_{J_m} 1_J=0$ because $D_{J_m} 1_{J_m} = 0$. Consequently, the vector $-\tilde S_{J_m}1_{J-1}$ equals the first column of $S_{J_m}$, where $\tilde S_{J_m}\in\mathbb R^{(J-2)\times (J-1)}$ is obtained by removing the first column of $S_{J_m}$. Thus, $[0_{J-2},\tilde S_{J_m}]A_\J = [-\tilde S_{J_m}1_{J-1},\tilde S_{J_m}]=S_{J_m}$. This leads to $\theta_\J^T  P_{J_m} \theta_\J = (S_{J_m} \theta_\J)^T S_{J_m} \theta_\J = (\tilde S_{J_m} \tilde \theta_\J)^T \tilde  S_{J_m} \tilde \theta_\J = \tilde \theta_\J^T  \tilde P_{J_m} \tilde \theta_\J $. Therefore, by rewriting \eqref{eqn:psplineprior} with $\text{rank}(P_\J)=\text{rank}(\tilde{P}_\J)$, we obtain the equivalent expression in \eqref{eqn:psplineprior_modified}.
\end{proof}

\begin{proof}[Proof of Lemma~\ref{lmm:sigma}]
	By Markov's inequality, 
	\begin{align*}
		\mathbb E_0\Pi\{\sigma>D_n\mid y\}\le D_n^{-2} \mathbb E_0\mathbb E(\sigma^2\mid y).
	\end{align*}
	By Fubini's theorem, 
	\begin{align*}
		\mathbb E_0\mathbb E(\sigma^2 \mid y) 
		&=\int \int  \mathbb E(\sigma^2\mid \J,\tau,\lambda, y) d \Pi(\J,\tau,\lambda\mid y) dP_0\\
		&=\int \mathbb E_0  \mathbb E(\sigma^2\mid \J,\tau,\lambda, y)  d\Pi(\J,\tau,\lambda\mid y) \\
		&\le \max_{\J:J\le n}\sup_{\tau\in(0,1)}\sup_{\lambda>0}\mathbb E_0  \mathbb E(\sigma^2\mid \J,\tau,\lambda,y)  .
	\end{align*}
	If $a_\sigma+n/2>1$, we obtain that 
	$$
	\mathbb E(\sigma^2\mid \J,\tau,\lambda,y)=\frac{2b_\sigma+n\hat\sigma_{\J,\tau,\lambda}^2}{2a_\sigma+n-2}\asymp\hat\sigma_{\J,\tau,\lambda}^2,
	$$
	where 
	$$\hat\sigma_{\J,\tau,\lambda}^2=\frac{1}{n}\left[y^T y - \frac{(y^T 1_n)^2}{n + \kappa^2} -
	y^T \tilde B_\J
	\Omega_{\J,\tau,\lambda}^{-1}
	\tilde B_\J^T y\right] =\frac{1}{n}y^T\left[ I_n - \frac{1_n1_n^T}{n + \kappa^2} -
	\tilde B_\J
	\Omega_{\J,\tau,\lambda}^{-1}
	\tilde B_\J^T \right]y. $$
	Let $ U_{\J,\tau,\lambda}= I_n - (1_n1_n^T)/(n + \kappa^2) -
	\tilde B_\J
	\Omega_{\J,\tau,\lambda}^{-1}
	\tilde B_\J^T $. Since $\mathbb E_0(y^T A y) = \sigma_0^2\text{tr}(A)+\mathbb E_0(y)^T A\mathbb E_0(y)$ for a symmetric matrix $A$,
	\begin{align}
		\begin{split}
			|\mathbb E_0(\hat\sigma_{J,\tau,\lambda}^2)-\sigma_0^2|&\le | n^{-1}\sigma_0^2 \text{tr}( U_{\J,\tau,\lambda})-\sigma_0^2|+n^{-1} F_0^T  U_{\J,\tau,\lambda} F_0 \\
			&\le n^{-1} \sigma_0^2 \text{tr}(I_n- U_{\J,\tau,\lambda})+n^{-1}  \lVert U_{\J,\tau,\lambda}\rVert_{\text{sp}} \lVert F_0\rVert_2^2,
			\label{eqn:exsig2}
		\end{split}
	\end{align}
	where $F_0=\mathbb E_0(y)\in\mathbb R^n$.
	Observe that 
	\begin{align}
		\begin{split}
	\text{tr}(I_n- U_{\J,\tau,\lambda}) &= \frac{n}{n + \kappa^2} +\text{tr}\!\left( 
	\tilde B_\J
	\Omega_{\J,\tau,\lambda}^{-1}
	\tilde B_\J^T \right) \\
	&= \frac{n}{n + \kappa^2} +\left(1 + \frac{\tau}{\lambda n} \right)^{-1}\text{tr}\!\left( 
	\Omega_{\J,\tau,\lambda}^{-1}\left(1 + \frac{\tau}{\lambda n} \right)
	\tilde B_\J^T\tilde B_\J \right).
			\end{split}
	\label{eqn:aa2}
\end{align}	
Plugging in the expression of $\Omega_{\J,\tau,\lambda}$, we obtain
		\begin{align*}
		\text{tr}\!\left( 
		\Omega_{\J,\tau,\lambda}^{-1}\left(1 + \frac{\tau}{\lambda n} \right)
		\tilde B_\J^T\tilde B_\J \right) &= \text{tr}\!\left( 
		I_{J-1} - \Omega_{\J,\tau,\lambda}^{-1}\left(\frac{1-\tau}{\lambda}\right) \tilde P_\J \right) \\
		&= (J-1)-\left(\frac{1-\tau}{\lambda}\right) \text{tr}\!\left( \Omega_{\J,\tau,\lambda}^{-1/2}\tilde P_\J  \Omega_{\J,\tau,\lambda}^{-1/2} \right)\\
		&\le J-1,
	\end{align*}
because $ \Omega_{\J,\tau,\lambda}^{-1/2}\tilde P_\J  \Omega_{\J,\tau,\lambda}^{-1/2} $ is positive semidefinite and the trace term is nonnegative. This implies that \eqref{eqn:aa2} is bounded by $J$.
It is also obvious that 
	$\lVert U_{\J,\tau,\lambda}\rVert_{\text{sp}} \le 1$ because $(1_n1_n^T)/(n + \kappa^2) +
	\tilde B_\J
	\Omega_{\J,\tau,\lambda}^{-1}
	\tilde B_\J^T$ is positive semi-definite. Hence, \eqref{eqn:exsig2} is bounded by a constant multiple of $J/n+1$ under \ref{asm:holder}. Putting everything together, we obtain that $\mathbb E_0\mathbb E(\sigma^2 \mid y)$ is bounded, and the assertion follows as soon as $D_n$ is increasing. 
\end{proof}

\begin{proof}[Proof of Lemma~\ref{lmm:eigen}]
Observe that $A_\J\theta_\J = (\theta_0,\tilde\theta_\J^T)^T$ with $A_\J$ in \eqref{eqn:linearmap}.
	Thus, we obtain that
	\begin{align}
		\label{eqn:eigenbo1}
		\begin{split}
			\Sigma_{\J,\sigma,\tau,\lambda}^{-1} &= \sigma^{-2} A_\J^T 
			\begin{pmatrix} \kappa^{-2} & 0_{J-1}^T \\
				0_{J-1} & \lambda^{-1} \!\left((1-\tau) \tilde P_\J + \tau n^{-1} \tilde B_\J^T \tilde B_\J \right)
			\end{pmatrix} A_\J \\
			&= \sigma^{-2} \!\left\{ \lambda^{-1} (1-\tau) A_\J^T 
			\begin{pmatrix} 
				0 & 0_{J-1}^T \\
				0_{J-1}& \tilde P_\J
			\end{pmatrix} 
			A_\J
			+
			A_\J^T
			\begin{pmatrix}
				\kappa^{-2} & 0_{J-1}^T \\
				0_{J-1}  & \tau/(\lambda n) \tilde B_\J^T \tilde B_\J
			\end{pmatrix}
			A_\J
			\right\}.
		\end{split}
	\end{align}
We showed that $[0_{J-2},\tilde S_{J_m}]A_\J = [-\tilde S_{J_m}1_{J-1},\tilde S_{J_m}]=S_{J_m}$ in the proof of Proposition~\ref{lmm:pspline}. Therefore, \eqref{eqn:eigenbo1} is equal to
	$$
	\sigma^{-2} \!\left\{ \lambda^{-1} (1-\tau) P_\J
	+
	A_\J^T R_\J^T R_\J A_\J
	\right\},
	$$
	where
	$$
	 R_\J = \begin{pmatrix} 1_n & \tilde B_\J \end{pmatrix} \begin{pmatrix} 1/(\kappa \sqrt{n}) & 0_{J-1}^T \\ 0_{J-1} & \sqrt{\tau/(\lambda n)} I_{J-1} 
	\end{pmatrix} .
	$$
	Since $[1_n , \tilde B_\J ] A_\J= B_\J$, we obtain that
	\begin{align*}
		\rho_{\min}(\Sigma_{\J,\sigma,\tau,\lambda}^{-1})&\ge \sigma^{-2}\left\{ \lambda^{-1}(1-\tau)\rho_{\min}(P_\J)+ \min(\kappa^{-2}, \tau/\lambda) \rho_{\min}(n^{-1}B_\J^T B_\J)\right\}, \\
		\rho_{\max}(\Sigma_{\J,\sigma,\tau,\lambda}^{-1})&\le \sigma^{-2}\left\{ \lambda^{-1}(1-\tau)\rho_{\max}(P_\J)+ \max(\kappa^{-2}, \tau/\lambda) \rho_{\max}(n^{-1}B_\J^T B_\J)\right\}.
	\end{align*}
First, we obtain $\rho_{\max}(n^{-1}B_\J^T B_\J)\le n^{-1}\lVert B_\J\rVert_{\text{sp}}^2\le \lVert B_\J\rVert_{\infty}^2\le 1$, where $\lVert B_\J\rVert_{\infty}$ is the usual $\infty$-norm (the maximum absolute row sum), which is 1 due to the sum-to-unity property. Clearly, $\rho_{\min}(P_\J)=0$. To find an upper bound for $\rho_{\max}(P_\J)$, observe that $P_\J = \sum_{m=1}^p S_{J_m}^T S_{J_m}$, and $S_{J_m}^T S_{J_m}$ and $S_{J_m} S_{J_m}^T$ share the identical nonzero eigenvalues. Since the eigenvalues of the Kronecker product of matrices are the product of the eigenvalues of each matrix (see, for example, Theorem 13.12 of \citet{laub2004matrix}),
we obtain
$$
\rho_{\max}(S_{J_m} S_{J_m}^T) = \rho_{\max}(I_{J_1}\otimes\dots\otimes I_{J_{m-1}}\otimes D_{J_m} D_{J_m}^T \otimes I_{J_{m+1}}\otimes\dots\otimes I_{J_p}) =\rho_{\max}( D_{J_m}D_{J_m}^T).
$$
The rightmost term is a symmetric banded Toeplitz matrix. Lemma~6 of \citet{gray2006toeplitz} shows that $\rho_{\max}(D_{J_m}D_{J_m}^T)$ is bounded by the sum of entries in the fourth row of Pascal's triangle, which is $16$. Therefore, $\rho_{\max}( P_\J)\le \sum_{m=1}^p \rho_{\max}( S_{J_m}^T S_{J_m})\le 16 p$.
\end{proof}

\bibliographystyle{apalike}
\bibliography{references.bib}

\end{document}